\documentclass[11pt]{article}
\usepackage{amsfonts}
\usepackage{amsmath}
\usepackage{amssymb}
\usepackage{enumerate}
\usepackage{natbib}
\usepackage{color}
\usepackage{array}
\usepackage{graphicx,epstopdf}
\usepackage[margin=1in]{geometry}
 \usepackage{float}
 \usepackage{comment}
\usepackage{rotating}
\usepackage{url}
\usepackage{subfigure}
 \usepackage[margin=1.5cm]{caption}
\usepackage{setspace}
\usepackage{hyperref}
\usepackage{breakurl}
\newtheorem{theorem}{Theorem}[subsection]

\newtheorem{corollary}[theorem]{Corollary}

\newtheorem{proposition}{Proposition}[section]

\newtheorem{remark}{Remark}[section]

\newenvironment{proof}[1][Proof]{\text{#1.} }{\ \rule{0.5em}{0.5em}}
\def\Q{{\mathbb Q}}        
\def\R{{\mathbb R}}        
\def\P{{\mathbb P}}        
\def\E{{\mathbb E}}        
\def\1{{\mathbf 1}}        

\numberwithin{theorem}{subsection}
\numberwithin{equation}{section}

\begin{document}
\title{Tracking VIX with VIX Futures:\\ Portfolio Construction and Performance\thanks{This draft is submitted for the Handbook of Applied Investment Research, edited by W. Ziemba et al.,  forthcoming, 2019.}}
\author{Tim Leung\thanks{Applied Mathematics Department,  University of Washington, Seattle WA 98195. Email: {timleung@uw.edu}. {Corresponding author}. } \and Brian Ward\thanks{Industrial Engineering \& Operations Research (IEOR) Department, Columbia University, New York, NY 10027. Email: {bmw2150@columbia.edu}.} }

\maketitle

\abstract{We study a series of static and dynamic portfolios of VIX futures and their effectiveness to track the VIX index.  We derive each portfolio using optimization methods, and evaluate its tracking performance from both empirical and theoretical perspectives. Among our  results, we show that static portfolios of different   VIX futures fail to track VIX closely.  VIX futures simply do not react quickly enough to movements in the spot VIX.  In  a discrete-time model, we design and implement  a  dynamic trading strategy that adjusts daily to optimally track VIX. The model is calibrated to historical  data and a simulation study is performed to understand the properties exhibited by the strategy. In addition, comparing to the volatility ETN, VXX, we find that  our dynamic strategy has a superior tracking performance.}\\

%
  
\newpage

\section{Introduction}
For all long-term investors it is very important to be able to effectively control the risk exposure of their portfolios through varying market conditions. Downside risk protection is particularly crucial. This motivates investors to consider different instruments and strategies to potentially hedge against market turbulence. One of the most widely recognized  measure of expected  market volatility is the CBOE Volatility Index (VIX). This is a desirable asset to hold as it has been observed empirically that volatility and market returns are anti-correlated, an effect known as \emph{asymmetric volatility}.\footnote{See \cite{blackAsymmetricVol} for the first use of this term.} Therefore, having exposure to VIX would protect the investor by providing some gains to offset market downturns. 

\begin{figure}[h]
	\begin{centering}
		\subfigure[April 1 -- Sep. 30, 2011]{\includegraphics[width=3in]{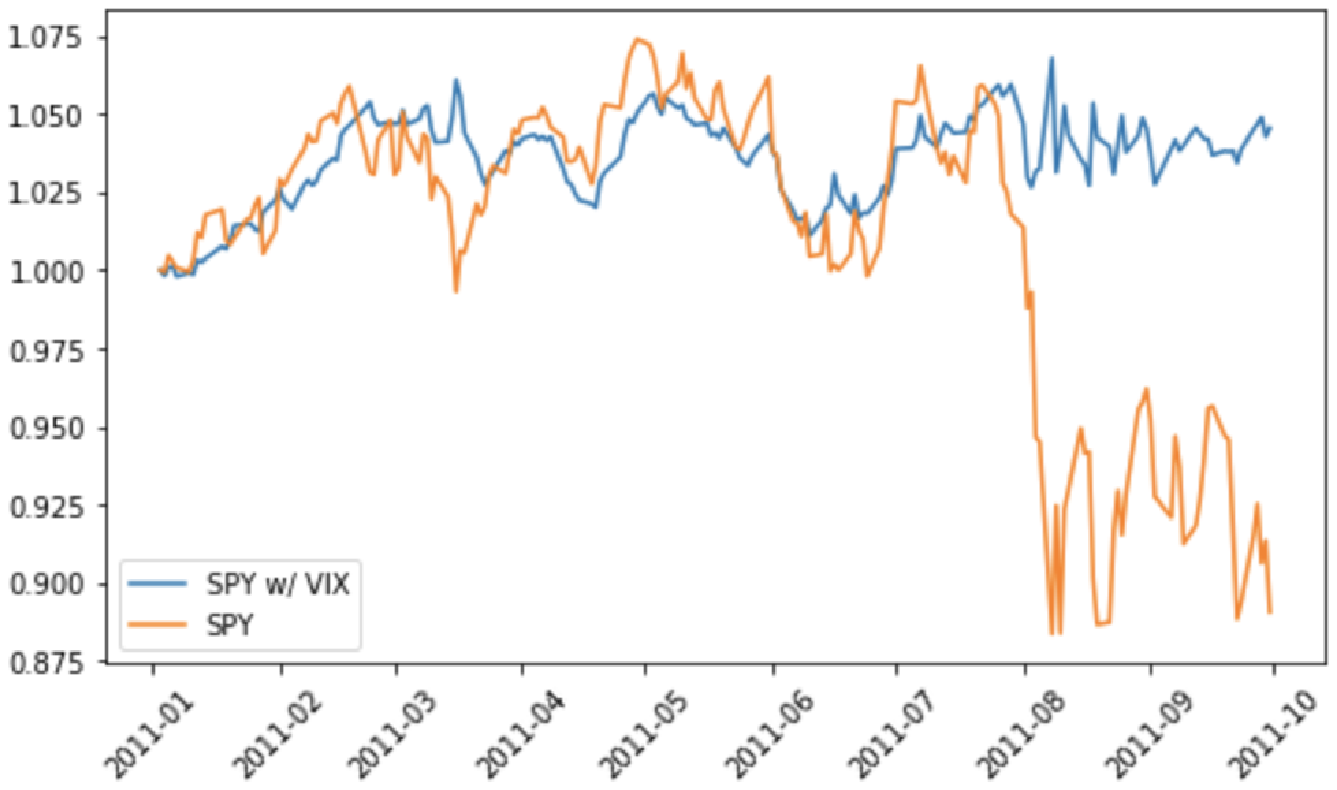}}
		\subfigure[Jan. 1 -- Dec. 31, 2014]{\includegraphics[width=3in]{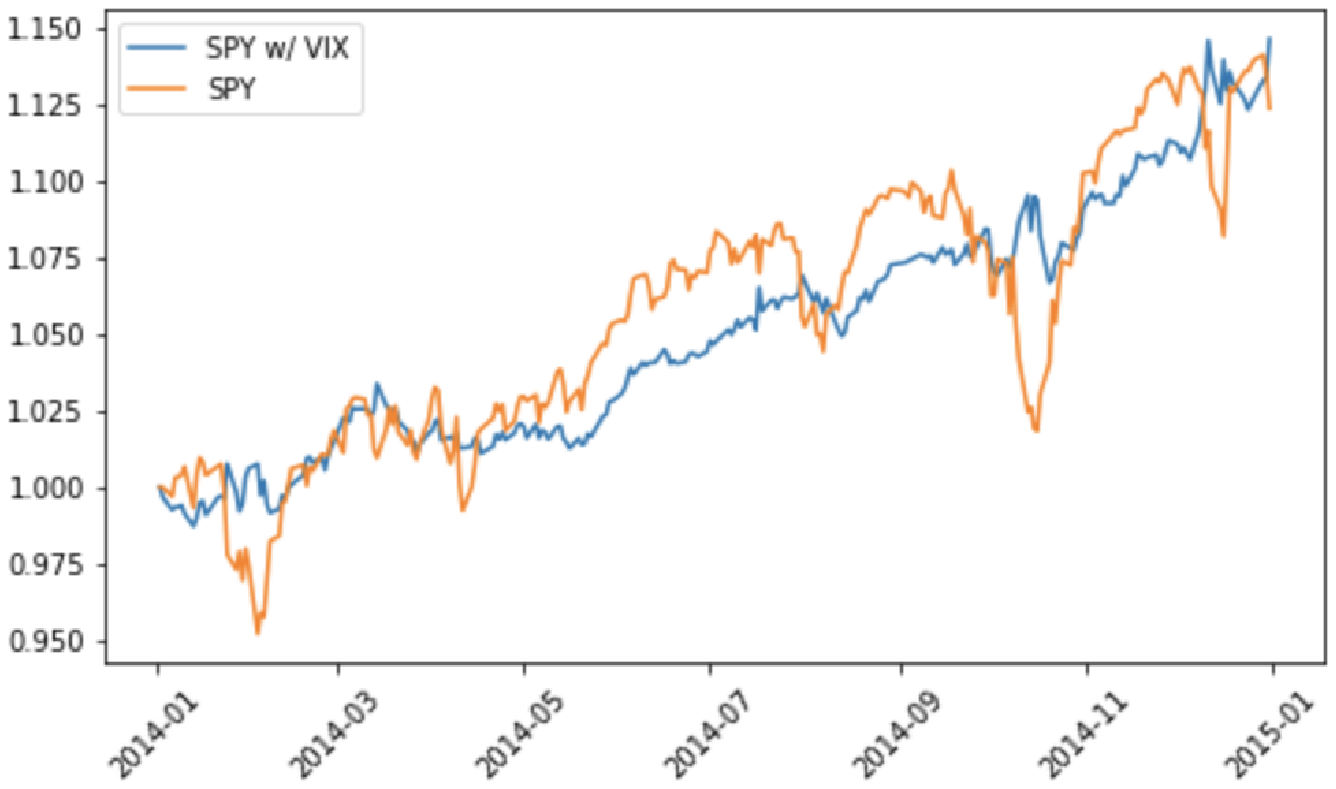}}
		\caption{\small{Historical portfolio values  during  (a) April 1, 2011 - Sep. 30, 2011, and (b) Jan. 1, 2014 - Dec. 31, 2014 for  a portfolio with 100\% investment in SPY, and another with 90\% wealth invested in SPY and  10\%  in VIX.}}
		\label{fig:VIX_SPY}
	\end{centering}
\end{figure}

To illustrate this, let us consider two portfolios: a portfolio fully invested in the SPDR S\&P 500 ETF (SPY), and another portfolio  with 90\% wealth in SPY and  10\%  wealth in VIX, assuming that VIX is tradable.\footnote{Also discussed in the authors' companion paper, \cite{wardIDX}.} Figure \ref{fig:VIX_SPY}(a) corresponds to the period of 2011 when the  U.S. suffered a credit rating downgrade by Standard and Poor's. News of a negative outlook by S\&P of the U.S. credit rating broke on April 18th, 2011.\footnote{See New York Times article: \url{http://www.nytimes.com/2011/04/19/business/19markets.html}.}  The SPY position would go on to lose about 10\%  with a volatile trajectory for a few months past the official downgrade on August 5th, 2011.\footnote{See  \url{http://www.nytimes.com/2011/08/06/business/us-debt-downgraded-by-sp.html}.} In contrast, a hypothetical portfolio with a 90/10 mix of SPY and VIX would be stable through the downgrade and end up with a positive return.  In Figure \ref{fig:VIX_SPY}(b),  we observe that  the same pair of portfolios earned roughly the same 15\% return in 2014. However, with SPY alone the portfolio is visibly more volatile  than the  portfolio with SPY and VIX. The reason is that  large drawdowns (for example on October 15th, 2014) were met by rises in VIX, creating a  stabilizing effect on the portfolio's value. 

Such a strategy is also beneficial for longer-term investment. If one invests 10\% of wealth in VIX and the remaining fraction in SPY on January 30, 2009, then the realized Sharpe ratio of holding this portfolio till December 27, 2017 is 1.16. By comparison, the SPY-only portfolio attains a Sharpe ratio of 1. The improvement is due mostly to a strong reduction in volatility from 16.07\% to 12.55\% and our calculations show further that investing up to  26\% of wealth in  VIX, with  remainder in SPY, would yield an even higher Sharpe ratio of 1.45. 

The benefits of having VIX exposure are obvious, but in reality VIX is not directly tradable. Instead, volatility exposure is achieved through the use of VIX futures or options, and a number of exchange-traded funds/notes (ETF/Ns). The most traded VIX-based ETN is the iPath S\&P 500 VIX Short-Term Futures ETN (ticker: VXX). However, as we will discuss, VXX fails to track VIX well and more generally, VIX ETF/Ns bring persistent negative returns. In the literature, a number of studies (\cite{DengMcCannWang,eraker2013,whaleyVolCost}) have also illustrated the negative returns associated with VIX futures and ETNs such as VXX. Nevertheless, VIX futures are also widely used for speculative trading and managed futures portfolios  (\cite{LeungLiLiZheng2015,Jiao2016,LeungYan2018,LeungYan2019}. As summarized by Bloomberg on January 29, 2019,  ``due to a structural quirk, the note lost 99 percent of its value over its life -- but it also democratized investor access to implied U.S. equity volatility. "

In this paper, we discuss the price dynamics of VIX futures and the underlying index, and construct static and dynamic portfolios of VIX futures for the purpose of tracking the index.  In addition to deriving and implementing the optimal tracking strategies, our main objective is to examine closely the effectiveness of the tracking portfolios.  For a thorough related study by the authors on discrete-time and continuous-time tracking of VIX and other indices involving futures/other derivatives, we refer to \cite{LeungWard2015,wardIDX,WardThesis}.

Our study  is structured as follows.  In Section \ref{sec:vix2}, we analyze the empirically observed return dependency of the VIX index and VIX futures. In Section \ref{sec:static_rep3}, we construct static portfolios of VIX futures for tracking VIX and investigate their tracking effectiveness over a long period of time. Our study illustrates a number of pitfalls in using static futures portfolios to attempt to track VIX, which motivates us to consider  a dynamic approach as discussed in Section \ref{sec:dt_model_vix}. Assuming a model that captures the mean-reverting dynamics of VIX, we derive the optimal dynamic replicating strategy that is adaptive to the daily fluctuations of VIX. The replicating strategy is implemented in Section \ref{sec:numerics}. Concluding remarks are provided in Section \ref{sect-conclude}.



\section{VIX Spot \&  Futures}\label{sec:vix2}
We begin by analyzing the price dynamics of the VIX Futures with respect to spot VIX. The historical price data for VIX futures  are obtained from Quandl.\footnote{Refer to \url{https://www.quandl.com/collections/futures/cboe} for documentation on the available CBOE data. One can search for specific contracts at \url{https://www.quandl.com/data/CBOE-Chicago-Board-Options-Exchange}.} We have validated the Quandl data against that directly from the CBOE.\footnote{CBOE publishes the historical data as well here: \url{http://cfe.cboe.com/data/historicaldata.aspx}.}  For the spot VIX data as well as the related ETNs, we use  Yahoo! Finance.\footnote{Quandl also has spot VIX data available, but it is sourced from Yahoo! Finance and we have validated the two datasets against each other.} 


After compiling, our dataset consists of the entire closing price history from March $26$, 2004 (the first day VIX futures began trading) through January $27$, 2017. However, we choose to analyze only the period from January $3$, 2011 (first trading date of 2011) through December $30$, 2016, which is a 6-year period from 2011 to 2016.  In our opinion, this amount of data is sufficient to avoid overfitting  and recent enough to understand the current dynamics among  VIX, futures, and VXX.  The entire sample period from 2011 to 2016 contains 1,510 total trading days of data. Moreover, on any given day of the sample set there are between 7 and 9 futures contracts available. In particular, there are 15 days with only 7 futures contracts available, 278 days with only 8 futures contracts available and 1,217 days with a full 9 months of contracts available. These contracts are  {always} consecutive months (starting with the 1-month contract) over this time period. In other words, when there are $N$ futures contracts available for trading, they always consist of the $N$ front months. As an example, if the current trading date were sometime early in January (before the January futures expiry), and 7 futures contracts were available for trading, then the maturities of the futures contracts would be the months of January through July, consecutively. These features of the dataset are consistent with CBOE protocol. They state they will currently list up to 9 near-term months for trading.\footnote{See \url{http://cfe.cboe.com/products/spec_vix.aspx}.} However, in our analysis, we eliminate the eighth and ninth month contract as it is not always the case that one can trade the eighth or ninth month contract. This allows us to use the full 1,510 days of data and avoid the need to eliminate the $15+278=293$ days with only 7 or 8 futures for trading. 

%

\subsection{Return Dependency}\label{sec:vix1dayreg}
We begin with a regression of the 1-day returns of VIX futures against the corresponding 1-day returns of VIX. The results are summarized in Table \ref{tab:VIX Regression Table}.  

\begin{table}[h]\centering\begin{small}
		\setlength{\extrarowheight}{2pt}
		\begin{tabular}{l r r c r} 
			\hline
			\text{Futures} & \text{Slope $\beta$} & \text{Intercept $\alpha$} & \text{$R^2$} & \text{$RMSE$} \\ \hline
			\hline
			1-month & 0.604 & $-$3.54$\cdot10^{-3}$ & 0.792 & 0.0242 \\
			2-month & 0.428 & $-$3.33$\cdot10^{-3}$ & 0.759 & 0.0189 \\
			3-month & 0.321 & $-$2.45$\cdot10^{-3}$ & 0.718 & 0.0158 \\
			4-month & 0.267 & $-$1.98$\cdot10^{-3}$ & 0.688 & 0.0141 \\
			5-month & 0.226 & $-$1.79$\cdot10^{-3}$ & 0.645 & 0.0131 \\
			6-month & 0.200 & $-$1.62$\cdot10^{-3}$ & 0.615 & 0.0124 \\
			7-month & 0.184 & $-$1.35$\cdot10^{-3}$ & 0.597 & 0.0118 \\
			\hline    \end{tabular}\end{small}
	\caption[Single Day Return Regressions for VIX Futures]{\small{A summary of the regression coefficients and measures of goodness of fit for regressing one-day returns of 1-month through 7-month futures on 1-day returns of spot VIX from Jan. 3, 2011 to Dec. $30$, 2016. The root mean squared error is defined by $RMSE:=\sqrt{\sum_{i=1}^n (r_i^{(V)}-\alpha-\beta r_i^{(j)} )^2/n},$ where $r_i^{(j)}$ is the 1-day return of the $j$th futures contract on trading day $i$, $r_i^{(V)}$ is the corresponding return of spot VIX, and $n$ is the number of trading days. The front seven month contracts are index by $j\in\{1, \ldots, 7\}$.}}\label{tab:VIX Regression Table}
\end{table}

In Table \ref{tab:VIX Regression Table}, we observe  the high $R^2$ values for all futures, suggesting that they are highly correlated with the spot. The  futures with shorter maturities have higher $R^2$ values. This can be explained by (i) the fact that futures prices tend to approach the spot price towards maturity, and (ii) that long dated contracts are less liquid than the short term contracts.  The slope coefficients are all statistically significant and less than $1$, which is intuitive as futures returns tend to  be less volatile than   spot returns. The negative intercepts, which are statistically significant, indicate that futures prices tend to fall even if the spot price does not move.  The reason lies in the term structure of VIX futures, which is as we will see is typically increasing in time-to-maturity. Even if the spot price does not move, the futures prices tend to decrease to match the spot price towards maturity, contributing to a negative intercept.

In Figure \ref{fig:correlationfigure2}, we plot the time series of spot VIX, the 1-month futures price (May-16 contract), and the  7-month futures price (Nov-16 contract) over the period from April $22$, 2016 (1 day after the expiration of the April-16 contract) to May $18$, 2016 (the expiration date of the May-16 contract). Notice that the spot, 1-month futures price and 7-month futures price all tend to move together. However, the moves in the spot are larger than moves in the futures prices. We also observe that the spot and futures prices often move in the same direction, but the futures prices do not move 1 for 1 with the spot. As one might expect, this effect is more pronounced for the 7-month contract, which barely moves over the period. For example, from trading day 4 to 5, we have a large up move in VIX of \$1.45 which is only met by an increase of \$1.00 by the 1-month contract and an increase of \$0.50 by the 7-month contract. 

\begin{figure}
	\centering
	\includegraphics[trim={0.5cm 1.4cm 0.5cm 1.4cm},clip,width=4.5in]{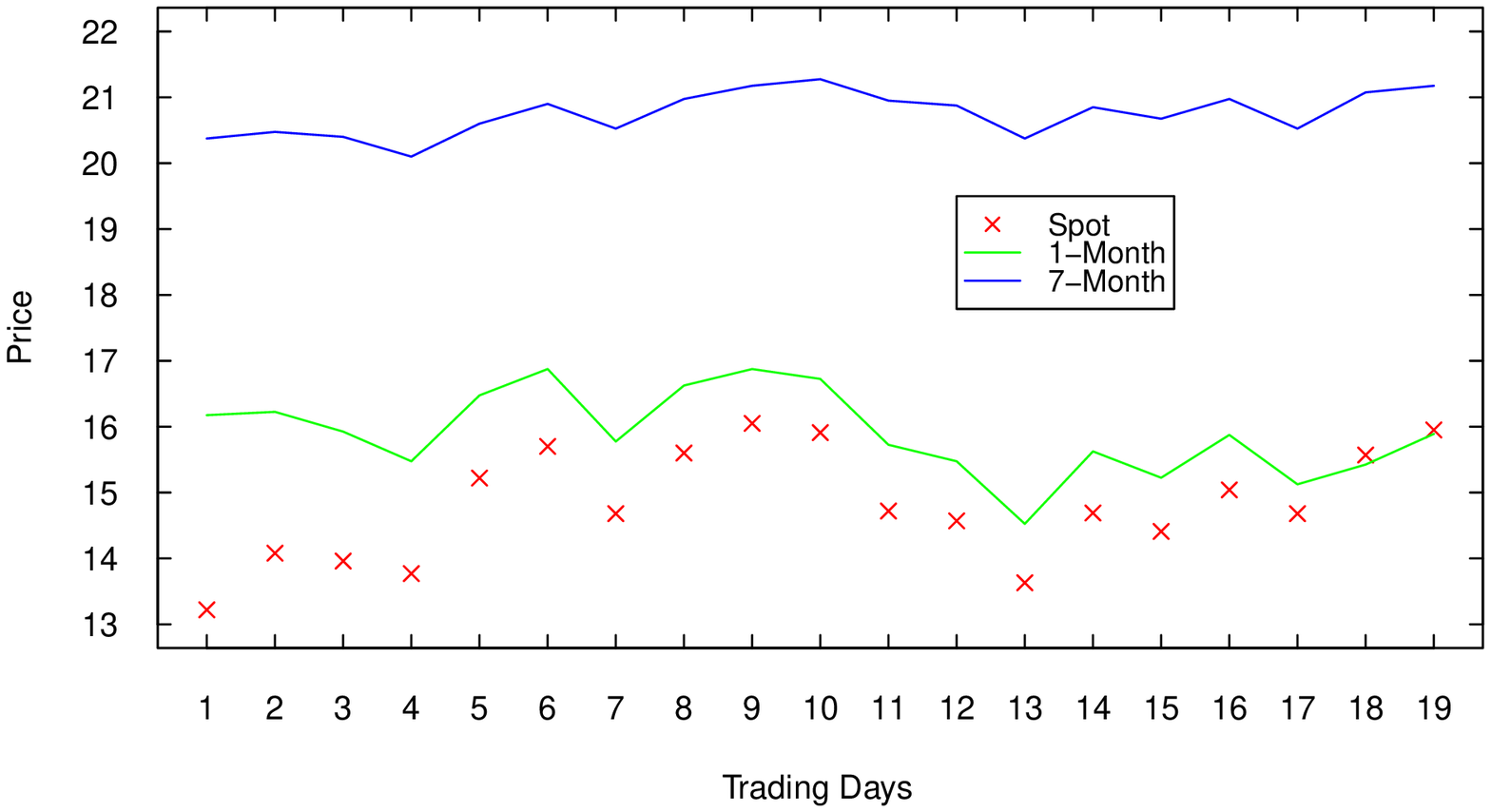}
	\caption [Price Evolution of VIX, 1-month Futures and 7-month Futures]{\small{Time series of VIX, 1-month futures price (May-16 contract) and 7-month futures price (Nov-16 contract) over the period from April $22$, 2016 to May $18$, 2016. The x-axis marks the trading day number, while the y-axis marks the price.}}\label{fig:correlationfigure2}
\end{figure}

Notice also that there is a slight discrepancy between the 1-month futures price at maturity  and   spot price. In theory, futures prices should converge to the spot price. However, market frictions, settlement rules\footnote{See \url{https://cfe.cboe.com/products/settlement_vix.aspx} for   VIX derivatives' settlement procedures.} and non-tradability of  the spot can substantially limit this convergence, as is seen in other futures markets, e.g. agricultural   futures (\cite{kevinNoConverge}). For the   non-convergence phenomenon of VIX futures, we refer to \cite{vixNoConverge} for an empirical study.

\begin{figure}
	\begin{centering}
		\subfigure[Jan.-June 2016]{\includegraphics[trim={0.5cm 1.4cm 0.5cm 1.4cm},clip,width=2.8in]{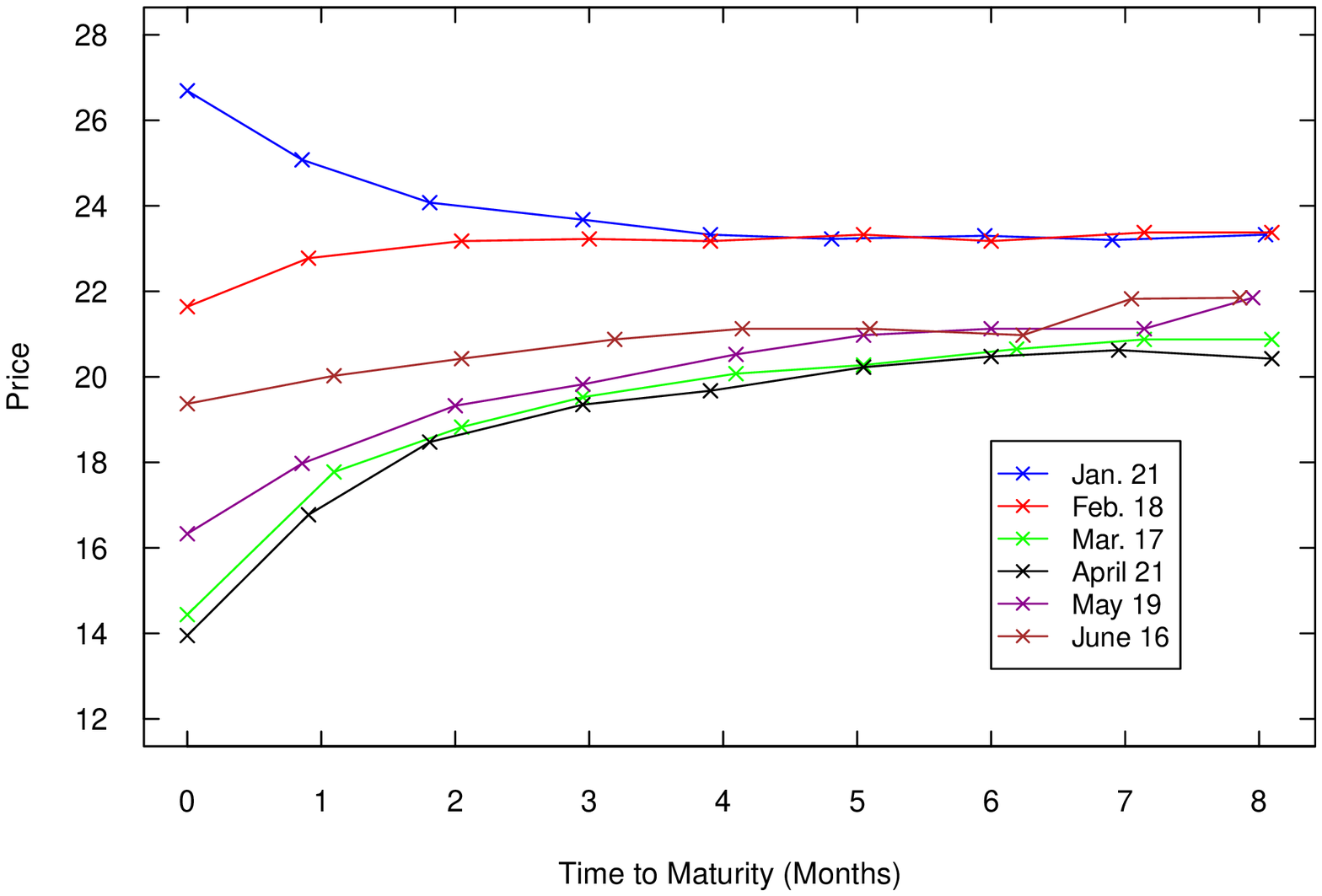}}
		\subfigure[Jan.-June 2009]{\includegraphics[trim={0.5cm 1.4cm 0.5cm 1.4cm},clip,width=2.8in]{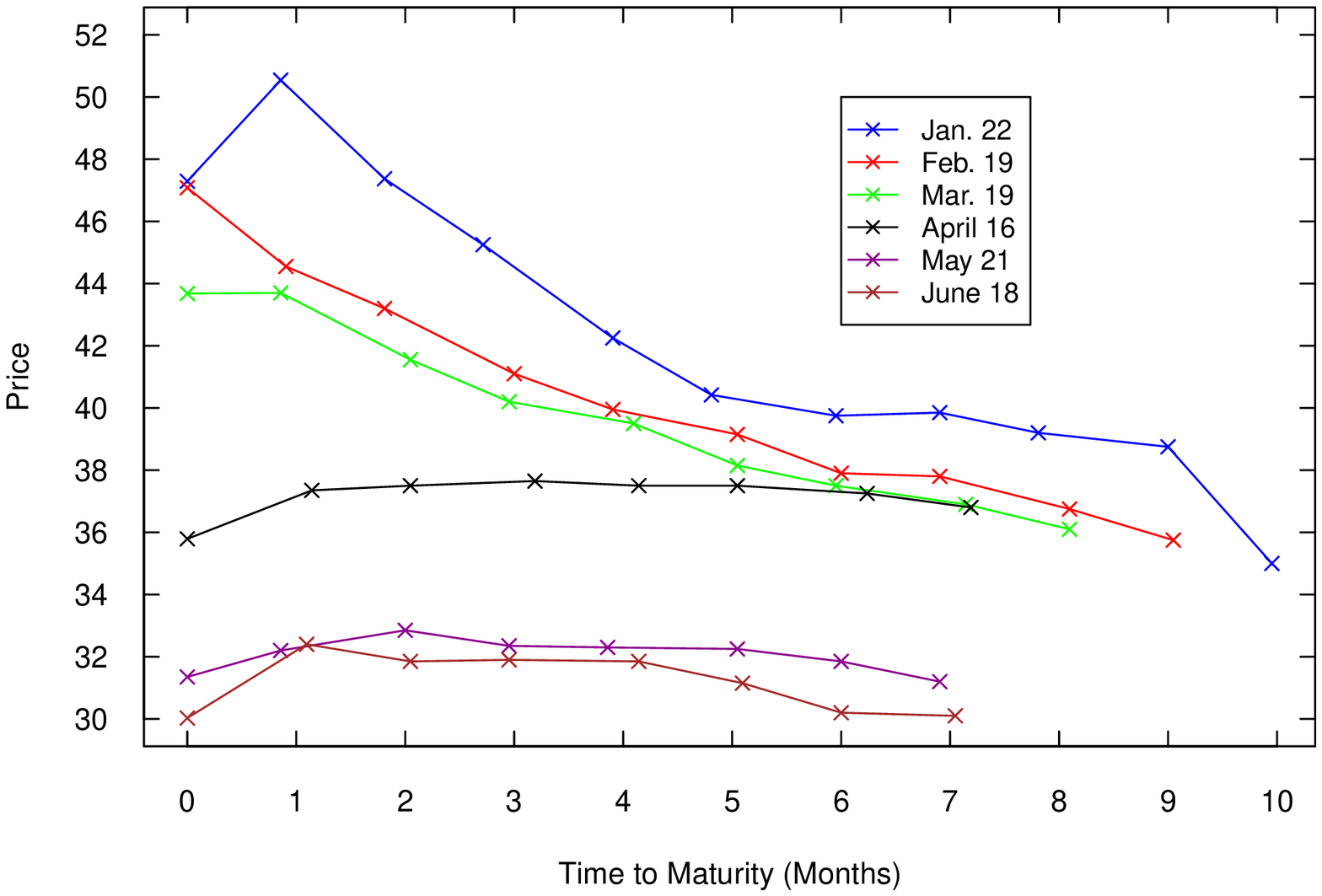}}
		\caption[Term Structures of VIX Futures]{\small{Term structures from Jan. to June in (a) 2016  and (b) 2009. The legend shows the dates upon which the term structures are constructed. Each date is the day after   that month's futures contract expires. The x-axis marks the time-to-maturity (in months) assuming each month is exactly 21 $(=252/12)$ days, while the y-axis marks the price.}}
		\label{fig:termstructure_vixchpt}
	\end{centering}
\end{figure}

  We plot the term structure of VIX futures as observed at several different time points throughout the first 6 months of 2016 (left) and 2009 (right) in Figure \ref{fig:termstructure_vixchpt}. The typical case for the VIX market is an increasing and concave futures curve (see left panel). However, in January 2016, VIX spiked as did expectations of future volatility.\footnote{One can find a discussion of trends in volatility in late January 2016 in the following article: \url{https://tickertape.tdameritrade.com/options/2016/01/volatility-high-early-2016-42727}.} Spikes like this can cause the term structure of VIX futures to invert, yielding a decreasing and convex futures curve. On the right panel of Figure \ref{fig:termstructure_vixchpt}, we plot the term structures in early 2009 where futures prices are at very high levels and the term structure shows different shape properties.  In the literature, there exist a number of models for VIX dynamics that yield the observed term structures (e.g. increasing concave, or decreasing convex). These include the Cox-Ingersoll-Ross (CIR)  model (see \cite{Grunbichler1996985} and \cite{futures_zhang}),  Ornstein-Uhlenbeck (OU) model and exponential OU model (see \cite{meanReversionBook}), and models with multiple stochastic factors and regimes (see \cite{MenciaSentana,Jiao2016}).
  


\subsection{Long-Term Dependency}\label{sec:vixlongterm}
When we conduct regressions of futures returns against spot returns over longer holding periods, new patterns emerge.  We use \emph{disjoint} intervals of various lengths when we compute the returns, meaning that for the longer horizons we have fewer data points, but even for the 30-day horizon we have approximately 50 data points. 

In Figure \ref{fig:VIX Regression Futures}, we plot the regressions of 1-month futures returns versus VIX returns for both 1-day returns (left) and 10-day returns (right), plotted on the same $x$-$y$ axis scale. The red x's mark pairs of returns, while the black line is the best fit line. One observes that the 10-day returns are much larger and consequently more volatile than the 1-day returns. However, for both holding periods, the futures returns are less volatile than the corresponding spot returns. To be precise, we find that the 1-day returns of the spot vary between -26.96\% and 50\%, while for the 10-day returns they vary between -39.76\% and 148.06\% (i.e. the longer holding period has more volatile returns). On the other hand for the 1-month futures, the 1-day returns vary between -20.81\% and 35.83\%, while the 10-day returns vary between -36.91\% and 88.89\% (i.e. futures returns are less volatile than the respective spot returns).  


\begin{figure}
	\centering
	\subfigure[1-Day Returns]{\includegraphics[trim={0.5cm 1.4cm 0.5cm 1.4cm},clip,width=2.8in]{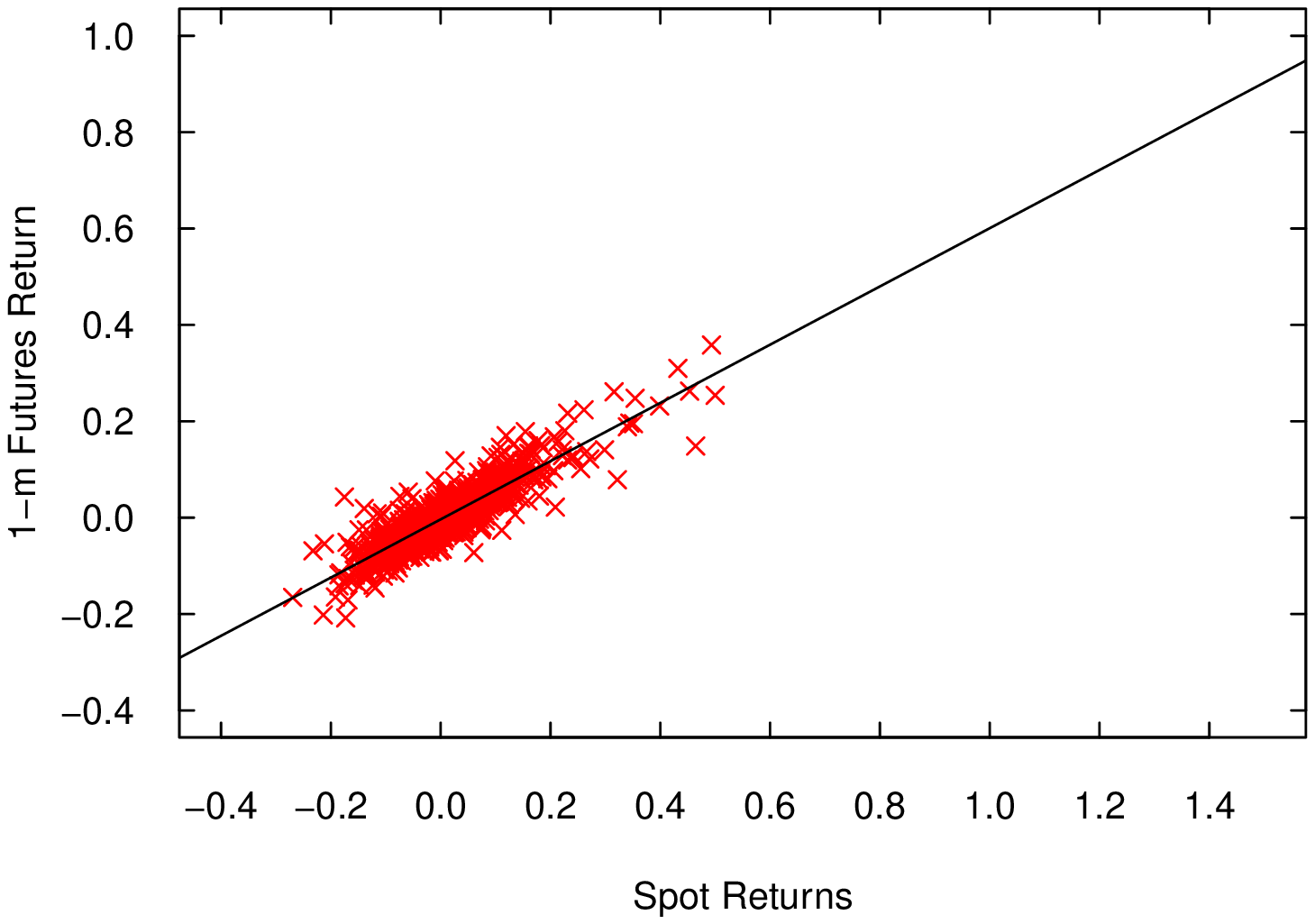}}
	\subfigure[10-Day Returns]{\includegraphics[trim={0.5cm 1.4cm 0.5cm 1.4cm},clip,width=2.8in]{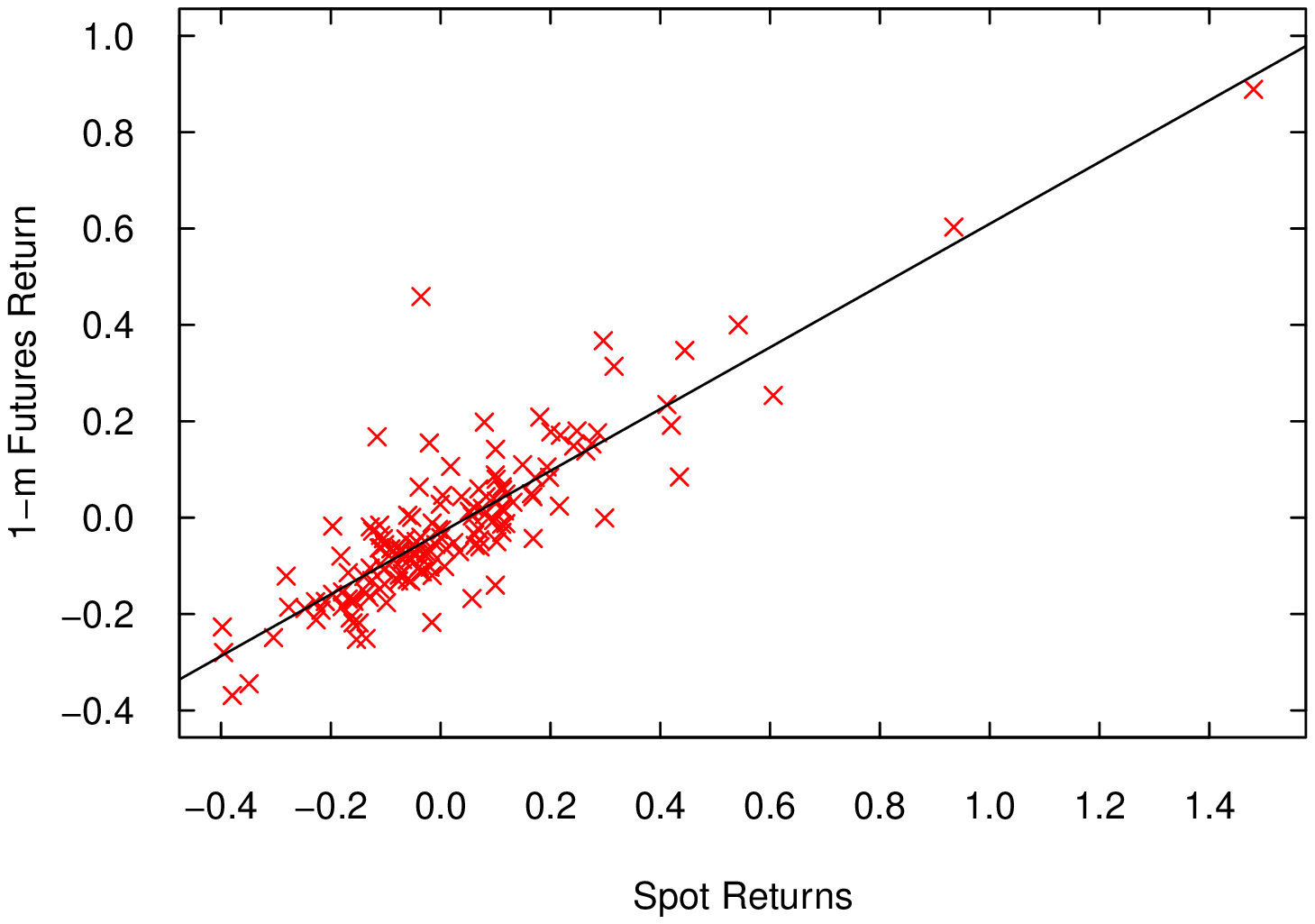}}
	\caption[Return Scatterplots for VIX Futures vs. VIX]{\small{Linear regressions of 1-month futures returns vs. spot VIX returns based on (a) 1-day returns and (b) 10-day returns. The x-axis marks the returns (in decimals) of spot VIX and the y-axis marks the returns (in decimals) of the 1-month futures.}}\label{fig:VIX Regression Futures}
\end{figure}

In the plot, the slope is slightly higher for the 10-day returns as compared to the 1-day returns. Moreover, for the 1-day returns the scatter plot is more tightly bound to the best fit line indicating a better fit than observed for the 10-day returns. However, this pattern does not hold in general. There is no discernible pattern in either the $R^2$ values or slope coefficients as the holding period is lengthend. This can be demonstrated by looking at a plot of $R^2$ or slope vs. holding period. We omit such a plot here as there is no clear pattern of increasing or decreasing predictive power (as measured by $R^2$) or leverage necessary to replicate the spot returns in the plot for any maturity contract. (Leverage can be measured by the reciprocal of the slope.) 


On the other hand, if we fix the holding period and increase the maturity of the contract, we see a decrease in predictive power and a decrease in slope. (Thus, an increase in leverage necessary to replicate the spot.) We demonstrate this for a number of different holding periods (1, 5, 10 and 15 days) in Table \ref{tab:holdingperiodslopes2}. In the top half, we give the slope coefficients for the regressions of futures returns against spot returns and in the bottom half, we display the $R^2$ values. By looking across the rows in either half, one notices a decrease in the reported statistic. This reaffirms our earlier observation that the spot is more closely tracked by short-term futures than long term futures. Indeed, over a 15 day period, the results indicate that about 1.6x leverage is required to track the spot with 1-month futures vs. 7.2x leverage for the 7-month futures. These implied leverage values are obtained by computing the reciprocals of the slopes (0.622, and 0.139, respectively). The leverage for the 7-month futures is quite substantial and is likely not be feasible in the marketplace due to trading costs or exchange limits.


\begin{table}[H]\centering \begin{small}
		\setlength{\extrarowheight}{2pt}
		\begin{tabular}{l c c c c c c c c}
			\hline
			&\text{Days} & \text{1-Mon} & \text{2-Mon} & \text{3-Mon} & \text{4-Mon} & \text{5-Mon} & \text{6-Mon} & \text{7-Mon} \\ \hline
			\hline   
			{Slope} 
			& 1 & 0.604 & 0.428 & 0.321 & 0.267 & 0.226 & 0.200 & 0.184 \\
			& 5 & 0.668 & 0.471 & 0.358 & 0.291 & 0.244 & 0.219 & 0.199 \\
			& 10 & 0.641 & 0.424 & 0.320 & 0.254 & 0.206 & 0.185 & 0.167 \\
			& 15 & 0.622 & 0.373 & 0.283 & 0.218 & 0.171 & 0.153 & 0.139 \\
			\hline \hline
			\text{$R^2$} 
			& 1 & 0.792 & 0.759 & 0.718 & 0.688 & 0.645 & 0.615 & 0.597 \\
			& 5 & 0.832 & 0.825 & 0.779 & 0.738 & 0.697 & 0.673 & 0.649 \\
			& 10 & 0.750 & 0.750 & 0.691 & 0.630 & 0.562 & 0.540 & 0.519 \\
			& 15 & 0.663 & 0.639 & 0.574 & 0.488 & 0.395 & 0.376 & 0.352 \\
			\hline    \end{tabular}
	\end{small}     
	\caption[Multiple Day Return Regressions for VIX Futures]{\small{A summary of the slopes and $R^2$s from the regressions of futures returns versus VIX returns over different holding periods.}}\label{tab:holdingperiodslopes2}
\end{table}

Although not immediately discernible, in Figure \ref{fig:VIX Regression Futures}, the intercept is more negative for the 10-day returns as compared to the 1-day returns. This confirms a property we have already discussed: futures tend to underperform and lose money relative to the spot returns. The more negative intercept indicates that this underperformance worsens over longer horizons. To see this more generally, in Figure \ref{fig:alphas}, we plot the intercepts for the regressions of returns of 1-month futures (black), 3-month futures (red) and 6-month futures (blue) across many different holding periods (from 1 day up to 30 days). It is quite clear that as the holding period is lengthened, the intercept becomes increasingly more negative. All intercepts (as well as the slopes from before) reported here are statistically significant at the 1\% significance level. Thus, a statistically significant discrepancy exists and continues to worsen as holding period is lengthened. 

\begin{figure}[H]
	\centering
	\includegraphics[trim={0.5cm 1.6cm 0.7cm 1.5cm},clip,width=3.7in]{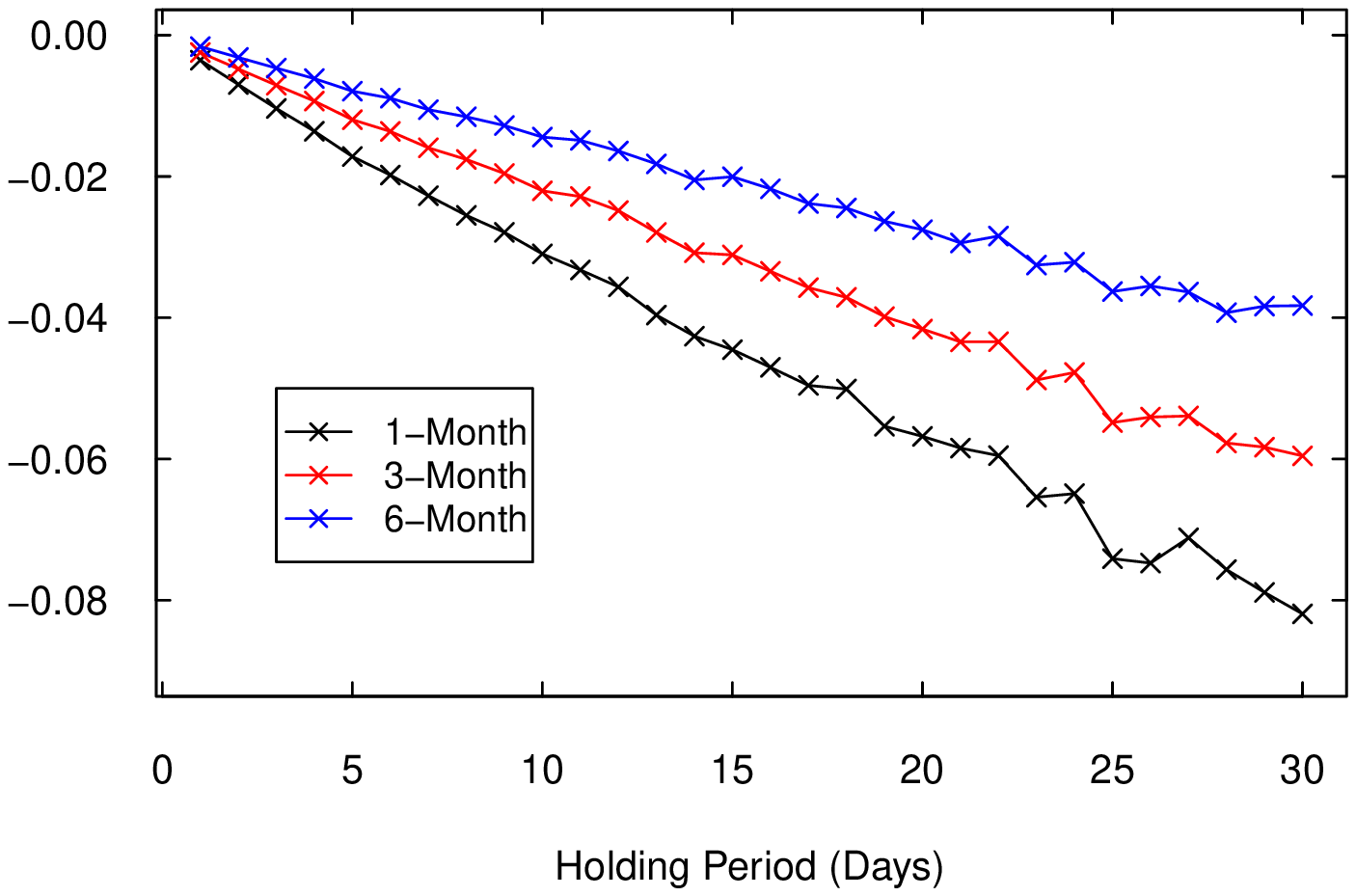}
	\caption[Regression Intercepts vs. Holding Period for VIX Futures]{\small{The intercepts from regressions of returns from a rolling futures position against spot returns, plotted against the holding period measured in trading days. From top to bottom: 6-month, 3-month, and 1-month futures.}}\label{fig:alphas}
\end{figure}

\section{Static Replication Strategies}\label{sec:static_rep3}
In this section, we investigate the static replication of VIX using  futures. Portfolios based on two different optimization criteria are considered. In Section \ref{sec:vixstaticrep}, we consider replication by optimizing the weights of the portfolio to match the physical dollar price of VIX over the training period. Due to the mean-reverting nature of VIX, the optimal portfolio is effectively all cash as this will keep the portfolio roughly near the mean level of VIX through the in-sample period. This motivates optimizing the portfolio to match the returns as closely as possible. The performance is improved upon with this change of criterion, but we demonstrate that it simply is difficult to track VIX with a static portfolio of futures. Unless otherwise noted,  the training or in-sample period is  the 5-year period from 2011 to 2015, and the test or out-of-sample period is the year 2016.  

\subsection{Physical Replication by Static Futures Portfolios}\label{sec:vixstaticrep}
In this section, we consider static positions including 1 or more futures contracts. By a static position in the futures contract, we mean a position that is constantly rolled over into the next futures contract. For example, a static investment of $\$P_0$ in 1-month futures would proceed as follows. First, purchase $P_0/F_0$ units of the current 1-month contract and hold it until maturity. The position's value at that time can be denoted by $P_1$. Then, one uses this final value to purchase what was previously (e.g. at time 0) the 2-month futures, but is now the 1-month futures. If this futures price is denoted $F_1$, then one purchases $P_1/F_1$ units of the new 1-month contract and hold \emph{that} contract until maturity. This process is continued indefinitely at each maturity date. 


There will be no cash withdrawals/injections at the roll periods (e.g. to keep the number of units of futures constant.) This makes the position completely self-financing, though the number of achievable units of futures will vary over time (i.e. it is completely possible that $P_0/F_0\neq P_1/F_1$). The static position described here, starting on January $3^\text{rd}$, 2011 with \$100 will be called the value of the 1-month futures contract. One can analogously define the 2-month, 3-month, etc. values as positions that always maintain all money in 2-month, 3-month, etc. contracts. 

The portfolio value on day $j$ is denoted by $P_j$ and consists of various VIX futures contracts. The value of spot VIX on day $j$ is denoted by $V_j$. The precise weightings amongst the futures will be chosen so as to minimize the sum of squared errors:
\begin{equation}\label{eq:vixSSE}
	SSE=\sum_{j=1}^n(V_j-G_j)^2,
\end{equation}
where $n$ is the number of trading days in the in-sample period. 

Let $k$ be the number of futures contracts being considered and $\textbf{w} := (w_0, \ldots, w_{k})$ be the real-valued vector of portfolio weights. In particular,  $w_{0}$ represents the weight given to the money market account.\footnote{The money market account yields overnight LIBOR. The data was obtained from the Federal Reserve Bank of St. Louis (FRED); see \url{https://fred.stlouisfed.org/series/USDONTD156N}.} Next, denote by $\textbf{C}\in\R^{n\times k+1}$ a matrix containing as columns the historical values of the money market account and the various futures contracts. Finally, denote by $\textbf{d}\in\R^{n}$ a vector containing the historical values  of VIX. All prices come from the in-sample set and are normalized to start at \$100 on day $j=1$. With this notation and optimization criterion \eqref{eq:vixSSE}, we derive the following constrained least squares optimization problem:
\begin{equation}
\begin{aligned}
& \underset{\textbf{w} \in \R^{k+1}}{\text{min}}
& & \| \textbf{C}\textbf{w}-\textbf{d}\|^2 \\
& \text{s.t.}
& & \sum_{j=0}^{k} w_j=1.
\end{aligned}
\end{equation}

The normalization implies that an investor starting with $\$100$ will invest \$$100\cdot w_j$ into the $j$th futures contract and \$ $100\cdot w_0$ into the money market account. For simplicity of the subsequent analysis, we consider portfolios containing any of the four of the following contracts: 1-month, 2-month, 6-month and 7-month. Allowing for any number of futures in the portfolio this naturally yields $2^4-1=15$ (subtract 1 because we do not look at the portfolio using no futures as that is trivial) different potential portfolios. Using all 7 available futures and their $2^7-1=127$ subsets would be untenable and obscure the analysis. The four contracts selected here represent the shortest and longest dated maturities that are available in the VIX futures market on every day of both the in-sample and out-of-sample period. 

The resuls of the optimization are given in Table \ref{tab:pfTable}. As above, $w_0$ always represents the weight on the money market account. All other $w_i$ are ordered so that if $i<j$, $w_i$ represents the weight on a nearer term contract than $w_j$. For example, in row 6 we are considering a portfolio of 1-month and 6-month futures. Thus, $w_1$ is the weight on 1-month futures, while $w_2$ represents the weight on 6-month futures. The Root Mean Squared Error (RMSE) has a standard definition in terms of the SSE defined in Equation \eqref{eq:vixSSE}. Precisely,
\begin{equation}\label{eq:rmseDEF}
	\begin{aligned}
		RMSE:=\sqrt{\frac{SSE}{n}},
	\end{aligned}
\end{equation} 
where $n=1,258$ for the in-sample period and $n=252$ for the out-of-sample period. Our choice of initial capital of \$100 allows us to interpret RMSE  as the average percentage deviation of the portfolio value from the spot price during the in-sample or out-of-sample period.


The RMSE values are quite large, indicating poor in-sample and out-of-sample tracking performance. To benchmark the performance we can compare the RMSE values to those achieved by VXX. The in-sample RMSE for VXX is 72.05\%, while the out-of-sample value is 21.97\%. The values indicate the futures based portfolios perform significantly better in-sample than does VXX. Out-of-sample, the futures based portfolios performs only slightly better in a few cases. 

\begin{remark}
The difference in performance in-sample vs. out-of-sample for VXX is driven mostly by the difference in length of the two data sets. Recall that VXX is an unlevered portfolio of VIX futures, which we have demonstrated in Section \ref{sec:vix2} does not react quickly to spot VIX price movements. Thus, once VXX has fallen away from VIX, it cannot return to the dollar value of VIX and additions to SSE continue to accumulate.
\end{remark}

Though the performance is somewhat better than VXX, the resulting strategies break with intuition. All portfolios are barely net long in the various futures contracts employed. This is seen prominently in the those with one or two futures only. For example, the 1-month only portfolio invests only 15.2\% of wealth in the value of the 1-month futures contract and the remaining 84.8\% is held in cash. This is due to the choice of SSE as an optimization criterion and that we measure SSE on the physical dollar price. This combination makes the optimization somewhat forward looking. It considers all at once the values VIX will take over the in-sample period. In particular, VIX is highly mean-reverting and thus, holding the portfolio value at the average level of VIX over the period is the best choice. Futures simply do not track the spot well so it is best to ignore any spikes and remain constant over time. 

Finally, in Figure \ref{fig:opt_pf} we plot the time series of spot VIX, the ETF VXX, and the optimal (across all combinations of portfolio components) out-of-sample portfolio. As discussed above, the optimal portfolio is relatively flat holding most of its value (91.7\%) and is barely net long in 1-month and 6-month futures (8.3\%). On the other hand VXX is more reactive to the spot price movements, but quickly diverges from spot VIX. Thereafter it is unable to recover and continues to fall away from spot VIX. By the end of the in-sample period, VXX has fallen to \$3.44, and for the out-of-sample period, it has fallen to \$29.89. In neither does samples does the optimal portfolio diverge in this manner. It instead stays roughly constant at VIX's average level over both periods. Neither portfolio appears a perfect surrogate for trading VIX and therefore we consider alternative methods to constructing portfolios of futures to track VIX in the next section.\\

\begin{table}[H]\centering \begin{small}
		\setlength{\extrarowheight}{2pt}
		\begin{tabular}{l r r r r r c c}
			\hline  
			\text{Futures} &$w_0$    &  $w_1$     &   $w_2$     &$w_3$    &$w_4$ 	& \text{in-RMSE} & \text{out-RMSE}\\
			\hline
			\hline
			{1-m}&0.848&  0.152 &          -  &       - &       - & 30.724 &23.345 \\
			{2-m}&0.857&  0.143 &          -  &       - &       - & 31.153 &26.074\\
			{6-m}&0.811&  0.189 &          -  &       - &       - & 31.026 &28.840\\
			{7-m}&0.777&  0.223 &          -  &       - &       - & 30.835 &28.616\\
			\hline
			{1-m,\,2-m}&0.932&  2.114 & $-$2.046 &       - &       - & 27.741 &19.521\\
			{1-m,\,6-m}&0.917&  0.345 & $-$0.262 &       - &       - & 30.576 &17.197\\
			{1-m,\,7-m}&0.840&  0.137 &  0.022 &       - &       - & 30.722 &23.819\\
			{2-m,\,6-m}&0.757& $-$0.203 &  0.446 &       - &       - & 30.960 &32.971\\
			{2-m,\,7-m}&0.669& $-$0.252 &  0.582 &       - &       - & 30.627 &33.518\\
			{6-m,\,7-m}&0.492& $-$2.146 &  2.654 &       - &       - & 29.801 &27.021\\
			\hline
			{1-m,\,2-m,\,6-m}&0.453&  3.770 & $-$5.615 & 2.392&       - & 23.404 &22.174\\
			{1-m,\,2-m,\,7-m}&0.341&  3.745 & $-$5.095 & 2.009&       - & 22.414 &24.742\\
			{1-m,\,6-m,\,7-m}&0.229&  1.979 &$-$11.066 & 9.858&       - & 22.960 &61.713\\
			{2-m,\,6-m,\,7-m}&0.359&  2.449 &$-$11.032 & 9.224&       - & 27.437 &37.218\\
			\hline
			{1-m,\! 2-m,\! 6-m,\! 7-m}&0.225&  3.292 & $-$3.151 &$-$5.548 &6.183 & 21.492 &44.222\\
			\hline    \end{tabular}
	\end{small} 
	\caption[Optimal VIX Tracking Futures Portfolios]{\small{Optimal portfolio weights/performance measures for portfolios of VIX futures for tracking the dollar value of VIX. Porfolios utilize any subset of the 1-month, 2-month, 6-month and 7-month contract.}}\label{tab:pfTable}
\end{table}

\begin{figure}[h]
	\begin{centering}
		\subfigure[In-Sample]{\includegraphics[trim={0.8cm 1.6cm 0.9cm 1.6cm},clip,width=4.5in]{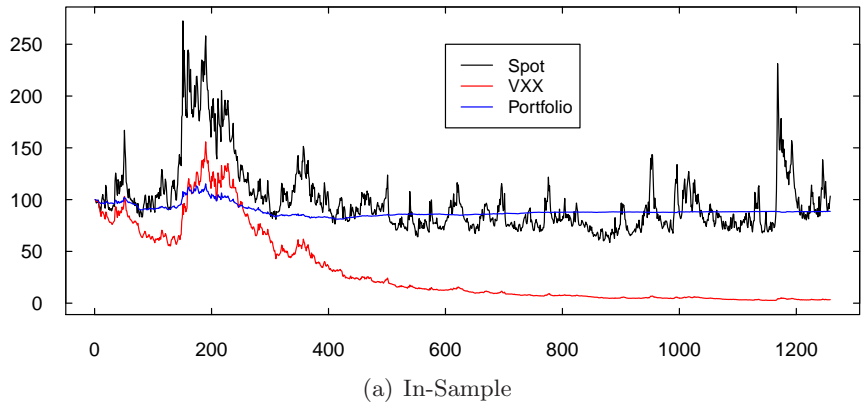}}
		\caption{\small{Time series of spot VIX, the ETF VXX and the optimal VIX dollar value tracking portfolio. The x-axis marks the trading day number, while the y-axis marks the price. All prices are normalized to start at \$$100$.}}
		\label{fig:opt_pf}
	\end{centering}
\end{figure} 

\subsection{Returns Replication with Static Portfolios}\label{sec:vixstaticretrep}

The results of Section \ref{sec:vixstaticrep} motivate considering a different tracking criterion. In particular, the tracking portfolios were all roughly constant and were barely net long VIX futures. Moreover, VXX had substantial discrepancies that were unrecoverable over a long trading period but it was at least somewhat reactive to VIX price movements. Our calculations show that VXX does well in tracking the returns of VIX relative to our portfolios. If we instead compute the RMSE value on the returns (rather than the physical price), we find in-sample VXX has a return RMSE of 4.63\% as compared to 7.51\% for our best (the 1-month, 6-month pair from before) portfolio. Those values are 4.55\%, and 7.21\%, respectively for the out-of-sample period. Since VXX tracks the returns well, it suggests building portfolios of VIX futures that also track the returns.


In the previous section our target for the tracking portfolios was the dollar value of VIX. If that were achievable at all points in time, the portfolios necessarily would have matched the VIX returns on each day as well. The converse need not be true: it is possible for the returns to be close to VIX returns while the physical values still diverge over time. This is clearly evidenced in the above paragraph whence VXX more closely matches the returns than our portfolio does, but more poorly tracks the physical price of VIX than the tracking portfolio does. In spite of this, we will consider return matching as the criterion for building portfolios in this section. 

The optimization problem remains mathematically the same as stated in Section \ref{sec:vixstaticrep}:
\begin{equation}
	\begin{aligned}
		& \underset{\textbf{w} \in \R^{k+1}}{\text{min}}
		& & \| \textbf{B}\textbf{w}-\textbf{y}\|^2 \\
		& \text{s.t.}
		& & \sum_{j=0}^{k} w_j=1.
	\end{aligned}
\end{equation}
The matrix $\textbf{B}$ contains as columns, the returns of the money market account and the various futures contracts under consideration. The vector, $\textbf{y}$ contains the historical returns of spot VIX. Final in-sample and out-of-sample performance is based on RMSE as defined in \eqref{eq:rmseDEF}, but this is done on the returns of the portfolio vs. the spot. 


The results of the optimization are reported in Table \ref{tab:pfTableRET}. By considering this new criterion, we obtain portfolios that are net long and are leveraged. The regression analysis of Section \ref{sec:vix2} suggested this would be the case. In fact, for all 15 portfolios, the weight on the money market account, $w_0$, is negative, indicating that borrowing is required. For example, the 6-month/7-month portfolio requires a leverage of about 3 to track the VIX returns (it borrows about 2x its value from the bank). This is also seen in the 6-month only portfolio and the 7-month only portfolios. 

For this tracking criterion, almost all portfolios outperform VXX in-sample (RMSE = 4.63\%) as well as out-of-sample (RMSE = 4.55\%). The exceptions are the 6-month only, 7-month only and 6-month/7-month portfolios. These two observations (extreme leverage and poor tracking performance) echo our previous discussion regarding the fact that the longer dated contracts are not as reactive to spot movements as the shorter dated contracts.

\begin{table}[h]\centering \begin{small}
		\setlength{\extrarowheight}{2pt}
		\begin{tabular}{l r r r r r c c}
			\hline  
			\text{Futures} &$w_0$    &  $w_1$     &   $w_2$     &$w_3$    &$w_4$ 	& \text{in-RMSE} & \text{out-RMSE}\\
			\hline
			\hline
			{1-m} & $-$0.317 & 1.317 & - & - & - & 3.620 & 3.491\\
			{2-m} & $-$0.751 & 1.751 & - & - & - & 3.929 & 3.685\\
			{6-m} & $-$2.018 & 3.018 & - & - & - & 4.926 & 4.659\\
			{7-m} & $-$2.178 & 3.178 & - & - & - & 5.020 & 4.832\\
			\hline
			{1-m, 2-m} & $-$0.499 & 0.921 &  0.578 & - & - & 3.532 & 3.430\\
			{1-m, 6-m} & $-$0.542 & 1.210 &  0.333 & - & - & 3.603 & 3.492\\
			{1-m, 7-m} & $-$0.503 & 1.237 &  0.265 & - & - &  3.610 & 3.491\\
			{2-m, 6-m} & $-$0.445 & 2.033 & $-$0.588 & - & - &  3.903 & 3.607\\
			{2-m, 7-m} & $-$0.380 & 2.050 & $-$0.669 & - & - &  3.896 & 3.598\\
			{6-m, 7-m} & $-$1.989 & 3.391 & $-$0.401 & - & - &  4.925 & 4.645\\
			\hline
			{1-m, 2-m, 6-m} & $-$0.319 & 0.905 & 0.769 & $-$0.355 & - & 3.522 & 3.400\\
			{1-m, 2-m, 7-m} & $-$0.244 & 0.904 &  0.810 & $-$0.470 & - & 3.515 & 3.396\\
			{1-m, 6-m, 7-m} & $-$0.463 & 1.215 &  1.269 & $-$1.021 & - & 3.592 & 3.494\\
			{2-m, 6-m, 7-m} & $-$0.379 & 2.040 &  0.193 & $-$0.854 & - & 3.896 & 3.600\\
			\hline
			{1-m,\! 2-m,\! 6-m,\! 7-m} & $-$0.237 & 0.910 & 0.770 & 0.594 & $-$1.036 & 3.511 & 3.405\\
			\hline    \end{tabular}
	\end{small}  
	\caption[Optimal VIX Return Tracking Futures Portfolios]{\small{Optimal portfolio weights/performance measures for portfolios of VIX futures for tracking VIX returns. Porfolios utilize any subset of the 1-month, 2-month, 6-month and 7-month contract. The reported RMSE values are for the returns of the portfolio.}}\label{tab:pfTableRET}
\end{table}

In terms of physical price replication, both VXX and the optimized portfolio continue to underperform, but that is not the goal of this section. Rather, we are concerned with return replication. Indeed the optimized portfolio is more reactive to spot movements and is no longer constant over the time period and the RMSE calculations demonstrate that the portfolio is tracking VIX returns well. However, with the leveraging, the portfolio has eroded so far in value that both in-sample and out-of-sample it has gone negative! The leveraging in combination with the fact that VIX futures tend to lose money even when the spot is unchanged causes the substantial loss in value. Interestingly, this portfolio is not overly leveraged (only about 1.24x).

Negative portfolio values are uncommon, but can occur when there is leveraging. As a simple example, suppose an investor begins with \$$100$ and constructs a 3x leveraged investment in an asset currently priced at \$$30$. Therefore, she borrows \$$200$ from the bank, at e.g. a constant (and annually compounded) interest rate of 5\% and purchases 10 units of the asset. Suppose that after 1 year's time, the asset's value has plummeted to \$$20$. The investor's 10 units are worth a total of \$$200$, but she owes $200\cdot1.05=210$ to the bank. The net value of her portfolio is therefore, $-10$, meaning a cash injection of \$$10$ is necessary just to liquidate the portfolio. A similar situation (though not as extreme) has occurred for the tracking portfolios here. The combination of leveraging along with the poor returns of VIX futures drove the portfolio to a negative value.

From this section and the previous one, it is quite evident that dynamic strategies are needed if one is to track the VIX with a portfolio of futures. Directly optimizing for the price trajectory yielded portfolios that did not react to VIX movements, and directly optimizing for return trajectory yield portfolios with negative values due to overleveraging. In the next section, we see a model-driven and dynamic strategy which appears to combat both pitfalls. 

\begin{figure}[h]
	\begin{centering}
		\subfigure[In-Sample]{\includegraphics[trim={0.8cm 1.6cm 0.9cm 1.6cm},clip,width=4.5in]{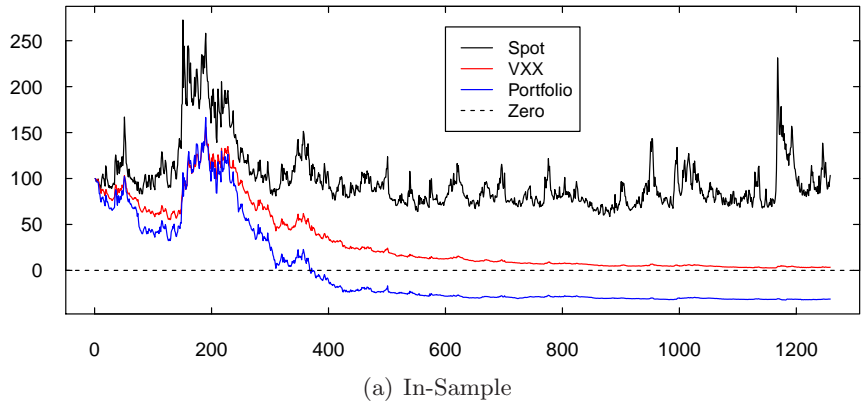}}
		\caption{\small{Time series of spot VIX, the ETF VXX and the optimal VIX return portfolio. The x-axis marks the trading day number, while the y-axis marks the price. All prices are normalized to start at \$$100$.}}
		\label{fig:opt_pf_ret}
	\end{centering}
\end{figure}

\section{Optimal Dynamic Replication}\label{sec:dt_model_vix}

In this section we develop an optimal strategy for replicating the returns of an index. The model is built on discrete-time dynamics for the index, futures on the index and the tracking portfolio. By working in discrete time, we will be able to capture the daily mark-to-market feature that is present in futures markets. We will then find the optimal dynamic weights to allocate capital amongst two futures contracts so as to replicate the daily returns of the index. 

\subsection{Futures Portfolio \& Price Dynamics}\label{Futures_Portfolio}

To maintain a futures position, the investor places a certain fraction of the exposure, e.g. 20\% of the dollar value, on margin with the clearing agency. This amount is called the \emph{initial margin}, while the account held with the clearing agency is called the \emph{margin account}. The margin account is \emph{marked-to-market} which means that if we denote the futures price for maturity $T_i$, as observed on day $j$ by $f_j^{(i)}$, the margin account receives a cash flow equal to $f_j^{(i)}-f_{j-1}^{(i)}$ for each unit of futures the investor is long at the end of day $j$. 

The investor is assumed to have no position limits when trading futures. Therefore, it is  possible for any portfolio weight  (fraction of the portfolio wealth) to be of any positive or negative value. The magnitude of the portfolio weight indicates the degree of leverage in the exposure.  Moreover, if these cash flows accumulate to a substantially negative value, the margin account will fall below a level called the \emph{maintenance margin}. In that case, the investor receives a \emph{margin call} and must replenish the account's funds back up to the initial margin. The account will otherwise be paid the risk free rate in exchange for the collateral. In our discrete-time model, we assume for simplicity that the full amount of the exposure is held on margin. This simplifies matters as there is no need to incorporate margin calls or maintenance margins in the model. 


Each time step will be one day, which we denote by $\Delta t$. Here, $\Delta t= {1}/{252}$ as there are $n=252$ trading days per year. Thus, the futures contracts are marked-to-market in accordance with the time step. Consequently, the investor must make all decisions for the next day's holdings based on information generated by the previous day's prices, and intraday positions are held  constant. The value of the investor's holdings is denoted by the discrete-time  process, $\left\{X_j\right\}_{j=0}^\infty$, where $X_0$ is the initial wealth. 

Further suppose there are $N$ futures contracts available for trading with maturities $T_1<...<T_N$. Then a portfolio holding a fraction of wealth, $w^{(i)}_j$, in the $i$th maturity, ($i=1,...,N$) will have $\frac{w^{(i)}_jX_{j}}{f_{j}^{(i)}}$ many units of futures contract of maturity $T_i$ on day $j$. Let the continuously compounded risk free rate be $r$ so that 1 dollar grows to $e^{r\Delta t}$ dollars the next day. From the above, it follows that if the portfolio of futures contracts is worth $X_{j}$ on day $j$, then it is worth
\begin{equation}\label{discrete_portfolio}
	\begin{aligned}
		X_{j+1}=X_{j}e^{r\Delta t}+\sum_{i=1}^{N} \frac{w^{(i)}_jX_{j}}{f_{j}^{(i)}}\left[f_{j+1}^{(i)}-f_{j}^{(i)}\right], \qquad \forall \,j\ge0\\
	\end{aligned}
\end{equation}
on day $j+1$. It follows that the daily return of the portfolio is given by
\begin{equation}\label{discrete_portfolio_ret}
	\begin{aligned}
		\frac{X_{j+1}}{X_j}-1=e^{r\Delta t}-1+\sum_{i=1}^{N} w^{(i)}_j\left[\frac{f_{j+1}^{(i)}}{f_{j}^{(i)}}-1\right], \qquad \forall \,j\ge0.\\
	\end{aligned}
\end{equation}
This is the sum of the risk free return for placing collateral with the exchange and the weighted average daily returns of the various futures contracts the investor holds.

The value of the index the investor tracks is denoted by the discrete-time stochastic process, $\left\{S_j\right\}_{j=0}^\infty$. Suppose the index satisfies the following equation (see below for motivation and discussion):
\begin{equation}\label{historical_vix_discrete}
	\begin{aligned}
		S_{j+1}-S_{j}={\mu}({\theta}-S_{j})\Delta t+g(j\Delta t,S_j)\sqrt{\Delta t}{Z}_{j+1}, \qquad \forall \,j\ge0,
	\end{aligned}
\end{equation}
where $\left\{Z_j\right\}_{j=1}^\infty$ are independent and identically distributed standard normal random variables under the historical measure $\P$. It follows that the daily return of the index is equal to
\begin{equation}\label{historical_vix_ret}
	\begin{aligned}
		\frac{S_{j+1}}{S_{j}}-1=\frac{\mu\theta\Delta t}{S_{j}}-\mu\Delta t+\frac{g(j\Delta t,S_j)}{S_{j}}\sqrt{\Delta t}{Z}_{j+1},  \qquad \forall \,j\ge0.
	\end{aligned}
\end{equation}
This  is motivated by the continuous-time mean-reverting model
	\begin{equation}\label{historical_vix_cts}
		\begin{aligned}
			dS_t = \mu(\theta-S_t)dt+g(t,S_t)dZ_t,
		\end{aligned}
	\end{equation}
	where $\mu,\theta$ are positive parameters, $g(\cdot,\cdot)$ is a generic local volatility function and $Z_t$ is a standard Brownian motion (SBM) under the historical measure $\P$. The drift $\mu(\theta-S_t)$ is positive (resp. negative) when  $S_t<\theta$ (resp. $S_t>\theta$) and thus $S_t$ tends to rise (resp. fall).
	That implies that $S_t$ mean-reverts to $\theta$. 
	
	By instantiating $g(\cdot,\cdot)$ with a particular function, one obtains many well-known mean-reverting models. For example, when $g(t,S_t)=\sigma$, the model is called the  {Ornstein-Uhlenbeck (OU) Model} (\cite{ouMod}) and when $g(t,S_t)=\sigma\sqrt{S_t}$, the model is called the  {Cox-Ingersoll-Ross (CIR) Model} (\cite{CIR85}).   For volatility  futures as an example,   \cite{Grunbichler1996985}  and  \cite{futures_zhang} model the S\&P500 volatility index (VIX) by  the CIR  process and provide a formula for the futures price.    Futures prices are computed under the risk-neutral measure, but for tracking and trading we also need to consider their dynamics under  the historical measure. As a result, the investor's optimal trading strategy depends on parameters from both measures.   
	
	We further assume that under the risk neutral measure, $\Q$, the index maintains the same mean-reverting property, but with different parameters. Namely, 
	\begin{equation}\label{riskNeutral_vix_cts}
		\begin{aligned}
			dS_t = \widetilde{\mu}(\widetilde{\theta}-S_t)dt+g(t,S_t)d\widetilde{Z}_t,
		\end{aligned}
	\end{equation}
	where $\widetilde{\mu},\widetilde{\theta}>0$, $g(\cdot,\cdot)$ is the same local volatility function as above\footnote{This is a standard property when we change measures: the drift changes, but the volatility remains the same.} and $\widetilde{Z}_t$ is a SBM under the risk neutral measure. We denote by $\lambda(t,S_t)$ the market price of risk that satisfies $dZ_t=d\widetilde{Z}_t-\lambda(t,S_t) dt$  so that 
	\begin{align}\label{lambdaCIR}
		\lambda(t,S_t) =\frac{{\mu}({\theta}-S_t)-\widetilde{\mu}(\widetilde{\theta}-S_t)}{g(t,S_t)}.
	\end{align}   
	This form of risk premium preserves the mean-reverting property  of the index under two measures.
	
	Under this model, the  price of a futures contract written on $S$ with maturity $T$ is
	\begin{align}\label{eq:chpt3_fTCIR}
		f^T_t := f(t,S;T) = \E^{\Q}\left[S_T|S_t=S\right] = (S-\widetilde{\theta})e^{-\widetilde{\mu}(T-t)} +\widetilde{\theta},
	\end{align}
	as long as the function $g(\cdot,\cdot)$ satisfies the integrability  condition
	\begin{align}\label{eq:chpt3_gcond}
		\E^\Q\left[\int_0^Te^{2\widetilde{\mu}\xi}g^2(\xi,S_\xi)d\xi\right]<\infty,\quad \forall\,\,\, T.
	\end{align} It follow from Ito's Formula that the dynamics of $f_t^T$ are given by
	\begin{align}
			df^T_t &= dS_t e^{-\widetilde{\mu}(T-t)} + \widetilde{\mu}e^{-\widetilde{\mu}(T-t)} (S_t-\widetilde{\theta}) dt\\
			&=  e^{-\widetilde{\mu}(T-t)} g(t,S_t)d\widetilde{Z}_t \label{futures_dynamics}
\\
			&= e^{-\widetilde{\mu}(T-t)} g(t,S_t)(\lambda(t,S_t)dt+dZ_t).\label{futures_dynamicsP}
		\end{align}
 	Equation  \eqref{futures_dynamics} indicates the dynamics under $\Q$, while \eqref{futures_dynamicsP} reflects the dynamics under $\P$.
 	
Following  \eqref{eq:chpt3_fTCIR}, we write the futures price with the discrete-time index, $j$, as
\begin{equation}\label{eq:futures_discrete}
	\begin{aligned}
		f_j^{(i)}=\widetilde{\theta}+(S_j-\widetilde{\theta})e^{-\widetilde{\mu}\left(T_i-j\Delta t\right)}, \quad \forall\, j\le \frac{T_i}{\Delta t}.
	\end{aligned}
\end{equation}
Here, $T_i$ is the maturity (measured in years) of the $i${th} futures contract and is a multiple of $\Delta t$. The above is defined $\forall\,j\le\frac{T_i}{\Delta t}$. To simplify notation, we have assumed that all contracts being considered have incepted at or before time $j=0$. For VIX futures, this assumption is valid if $N\le 7$ and the trading horizon is less than or equal 6 months. In the general case, one can simply update the set of available futures as well and still maintain the above equations for all $j\le \frac{T_i}{\Delta t}$.

To write down the return of each futures contract, we discretize \eqref{futures_dynamicsP}  to get 
\begin{equation}\label{historical_fut_discrete}
	\begin{aligned}
		f_{j+1}^{(i)}-f_{j}^{(i)} = e^{-\widetilde{\mu}\left(T_i-j\Delta t\right)}g(j\Delta t,S_j)\left(\lambda_{j} \Delta t+\sqrt{\Delta t}Z_{j+1}\right),\qquad\forall\,j\ge0,
	\end{aligned}
\end{equation}
where
\begin{align}\label{eq:lambda_disc}
	\lambda_j:=\lambda(j,S_j)=\frac{{\mu}({\theta}-S_j)-\widetilde{\mu}(\widetilde{\theta}-S_j)}{g(j\Delta t,S_j)},\qquad\forall\,j\ge0
\end{align}
is the discrete-time market price of risk. Let the time-to-maturity (in years) of the $i$th futures contract, on day $j$ be denoted by $D_j^{(i)}:=T_i-j\Delta t$. It follows that the daily return of the $i$th futures contract is
\begin{equation}\label{historical_fut_return}
	\begin{aligned}
		\frac{f_{j+1}^{(i)}}{f_{j}^{(i)}}-1 &= \frac{e^{-\widetilde{\mu}D_j^{(i)}}g(j\Delta t,S_j)}{\widetilde{\theta}+\left(S_{j}-\widetilde{\theta}\right)e^{-\widetilde{\mu}D_j^{(i)}}}\left(\lambda_j \Delta t+\sqrt{\Delta t}Z_{j+1}\right)\\
		&=\frac{g(j\Delta t,S_j)}{\widetilde{\theta}e^{\widetilde{\mu}D_j^{(i)}}+S_{j}-\widetilde{\theta}}\left(\lambda_j \Delta t+\sqrt{\Delta t}Z_{j+1}\right)\\
		&=B_{j}^{(i)}\left(\lambda_j \Delta t+\sqrt{\Delta t}Z_{j+1}\right),\qquad\forall\,j\ge0,\\
	\end{aligned}
\end{equation}
where we have defined
\begin{align}\label{eq:Bdef_chp3}
	B_{j}^{(i)}:=\frac{g(j\Delta t,S_j)}{\widetilde{\theta}e^{\widetilde{\mu}D_j^{(i)}}+S_{j}-\widetilde{\theta}},\qquad\forall\,j\ge0. 
\end{align}
Putting  \eqref{historical_fut_return} into  \eqref{discrete_portfolio_ret}, we now obtain the equation for the portfolio's return at each time step $j$:
\begin{equation}\label{discrete_portfolio_ret_plug}
	\begin{aligned}
		\frac{X_{j+1}}{X_j}-1=e^{r\Delta t}-1+\Delta t\sum_{i=1}^{N} w^{(i)}_jB_j^{(i)}\lambda_j+\sqrt{\Delta t}\sum_{i=1}^{N} w^{(i)}_jB_j^{(i)}Z_{j+1},\qquad\forall\,j\ge0.
	\end{aligned}
\end{equation}

\subsection{Optimal Tracking Problem}\label{Optimal_Tracking_Problem}
We now discuss a dynamic portfolio of  VIX futures designed to track the daily returns of  VIX. On each day $j$, the investor seeks to minimize  the conditional  expected squared deviation of the portfolio's return from a multiple $\beta\in\R$ of the index's return.
\begin{equation}\label{optimization_problem}
	\begin{aligned}
		\underset{w^{(i)}_j,i=1,...,N}{\min}\,\E^\P\left[\left(\frac{X_{j+1}}{X_{j}}-\beta\frac{S_{j+1}}{S_{j}}+\beta-1\right)^2\Big|\mathcal{F}_{j}\right],\\
	\end{aligned}
\end{equation}
where $\left\{\mathcal{F}_{j}\right\}_{j=0}^\infty$ is the discrete-time filtration representing information generated by prices observed as of day ${j}$. The leverage factor allows us to consider more general targets, though the value of $\beta=1$ was of particular interest in Sections \ref{sec:vix2} and \ref{sec:static_rep3}. The quantity $\beta-1$ comes from the fact that the objective criterion in Optimization Problem \eqref{optimization_problem} is a sqaured \emph{return} difference. That is, 
\begin{equation*}
	\begin{aligned}
		\frac{X_{j+1}}{X_{j}}-1-\beta\left(\frac{S_{j+1}}{S_{j}}-1\right)=\frac{X_{j+1}}{X_{j}}-\beta\frac{S_{j+1}}{S_{j}}+\beta-1.
	\end{aligned}
\end{equation*} 
Moreover, expectation is measured w.r.t. the historical measure, $\P$. We make this choice because the investor realizes cash flows in accordance with the historical (rather than risk neutral) measure. The risk neutral measure was only necessary to write down the discrete-time futures price equation.

%

The ETN, VXX is an exchange traded product that dynamically allocates wealth to two futures contracts. The rebalancing is based on the roll cycle of VIX futures and repeats itself during each cycle. In particular, at the beginning of the cycle, VXX holds 100\% of its wealth in the 1-month futures, and via daily rebalancing, reduces its holdings in 1-month futures and purchases 2-month futures. By maturity of the 1-month futures, VXX has sold off its entire position in 1-month futures and holds 100\% of its wealth in 2-month futures (which, at maturity \emph{is} the 1-month futures). The reduction in 1-month holdings by VXX is linear and deterministic. 

To reduce the complexity of Optimization Problem \eqref{optimization_problem} and facilitate comparison to VXX, we consider a particular case where $N=2$ with the added constraint that on each day $j$, we must have $w^{(i_1)}_j+w^{(i_2)}_j=1$. We do not require any other restrictions on the weights. Specifically, we do not set any constraints on the sign or size of the weights and in general we will find that leveraging the position in one futures contract is required to optimally track the VIX according to Optimization Problem \eqref{optimization_problem}. Note that $i_1$ and $i_2$ are parameters to the optimization problem and are not optimization variables themselves. They are left generic and we will consider several different pairs in the numerical implementation. 

The following proposition gives the complete solution to \eqref{optimization_problem} in our particular case:
\begin{proposition}\label{prop:prop}
	Under the discrete-time model \eqref{historical_vix_discrete} for the index $S$ with a portfolio of $N=2$ futures contracts on $S$ with maturities $T_{i_1}\neq T_{i_2}$ given by \eqref{discrete_portfolio}, the optimal strategy $(w_j^{(i_1)*}, w_j^{(i_2)*})$ associated with Optimization Problem \eqref{optimization_problem} is given by
	\begin{equation*}
		\begin{aligned}
			w_j^{(i_1)*}&=-\frac{\alpha_0\alpha_1+\nu_0\nu_1}{\alpha_1^2+\nu_1^2},\quad w_j^{(i_2)}&=1-w_j^{(i_1)*},
		\end{aligned}
	\end{equation*}
	with the optimal objective function value
	\begin{equation*}
		\begin{aligned}
			\frac{\left(\nu_1\alpha_0-\nu_0\alpha_1\right)^2}{\alpha_1^2+\nu_1^2},
		\end{aligned}
	\end{equation*}
	where
		\begin{align}
			&\alpha_0:=e^{r\Delta t}-1+\Delta tB_j^{(i_2)}\lambda_j-\beta\mu\Delta t\left(\frac{\theta}{S_{j}}-1\right), \label{a0}\\
			&\alpha_1:=\Delta t \lambda_j \left(B_j^{(i_1)}-B_j^{(i_2)}\right), \\
			& \nu_0:=\sqrt{\Delta t}\left(B_j^{(i_2)}-\frac{\beta g(j\Delta t,S_j)}{S_j}\right),\\
			&\nu_1:=\sqrt{\Delta t}\left(B_j^{(i_1)}-B_j^{(i_2)}\right).\label{nu1}
		\end{align}
\end{proposition}
\begin{proof}
	Combining     \eqref{historical_vix_ret} and \eqref{discrete_portfolio_ret_plug},  we compute the difference between the portfolio's return and   targeted multiple of the index's return:
	\begin{equation*}
		\begin{aligned}
			\frac{X_{j+1}}{X_{j}}-1-\beta\left(\frac{S_{j+1}}{S_{j}}-1\right)=e^{r\Delta t}-1+\Delta t\sum_{i=1}^{N} &w^{(i)}_jB_j^{(i)}\lambda_j+\sqrt{\Delta t}\sum_{i=1}^{N} w^{(i)}_jB_j^{(i)}Z_{j+1}\\
			&-\beta\mu\Delta t\left(\frac{\theta}{S_{j}}-1\right)-\frac{\beta g(j\Delta t,S_j)}{S_{j}}\sqrt{\Delta t}{Z}_{j+1}.
		\end{aligned}
	\end{equation*}
	This can be expressed in the affine form $\phi_0+\phi_1 Z_{j+1}$, where 
	\begin{equation*}
		\begin{aligned}
			\phi_0:=e^{r\Delta t}-1+\Delta t\sum_{i=1}^{N} w^{(i)}_jB_j^{(i)}\lambda_j-\beta\mu\Delta t\left(\frac{\theta}{S_{j}}-1\right),
		\end{aligned}
	\end{equation*}
	and 
	\begin{equation*}
		\begin{aligned}
			\phi_1 := \sqrt{\Delta t}\sum_{i=1}^{N} w^{(i)}_jB_j^{(i)}- \frac{\beta g(j\Delta t,S_j)}{S_j}\sqrt{\Delta t}.
		\end{aligned}
	\end{equation*}
	Both $\phi_0$ and $\phi_1$ are stochastic and depend on $j$, $S_j$, along with other model parameters. However, they are both measurable with respect to $\mathcal{F}_{j}$. Using  this with the fact that $Z_{j+1}$ is independent of $\mathcal{F}_{j}$, we conclude that, conditional on $\mathcal{F}_j$, the return difference is \emph{normally} distributed mean $\phi_0$ and variance $\phi_1^2$, and its conditional second moment is given by
	\begin{equation*}
		\begin{aligned}
			\left(e^{r\Delta t}-1+\Delta t\sum_{i=1}^{N} w^{(i)}_jB_j^{(i)}\lambda_j+\beta\mu\Delta t\left(1-\frac{\theta}{S_{j}}\right)\right)^2+\left(\sqrt{\Delta t}\sum_{i=1}^{N} w^{(i)}_jB_j^{(i)}- \frac{\beta g(j\Delta t,S_j)}{S_j}\sqrt{\Delta t}\right)^2.
		\end{aligned}
	\end{equation*}
	Now, we set $N=2$ and  $w\equiv w_j^{(i_1)}$ so that $w_j^{(i_2)}=1-w$. Then Optimization Problem \eqref{optimization_problem} is equivalent to
	\begin{equation}\label{opt_problem_simple}
		\begin{aligned}
			\underset{w\in\R}{\min}\,\left(\alpha_0+\alpha_1 w\right)^2+\left(\nu_0+\nu_1 w\right)^2,
		\end{aligned}
	\end{equation}
	where $\alpha_1$, $\alpha_0$ $\nu_1$ and $\nu_0$ are defined in \eqref{a0}--\eqref{nu1}. As with $\phi_0$ and $\phi_1$, all four of these coefficients are stochastic and depend on $j$, $S_j$ as well as the parameters of the model and futures contracts, but we suppress the dependence here.
	
	The square of a linear function is always convex and since the sum of convex function is convex, the overall optimization problem is convex. Therefore, the first-order condition, which  is necessary and sufficient for global optimality when solvable, is given by
	\begin{equation*}
 2\left(\alpha_0+\alpha_1 w\right)\alpha_1+2\left(\nu_0+\nu_1 w\right)\nu_1=0,
 	\end{equation*}which gives the optimal strategy
	\[ w^*=-\frac{\alpha_0\alpha_1+\nu_0\nu_1}{\alpha_1^2+\nu_1^2}.\]
	In order for this critical point to exist, we require $\alpha_1\neq0$ or $\nu_1\neq0$. Both are the case so long as $B_j^{(i_1)}\neq B_j^{(i_2)}$, which is equivalent to
	\begin{equation*}
		\begin{aligned}
			\frac{g(j\Delta t,S_j)}{\widetilde{\theta}e^{\widetilde{\mu}D_j^{(i_1)}}+S_{j}-\widetilde{\theta}}\neq\frac{g(j\Delta t,S_j)}{\widetilde{\theta}e^{\widetilde{\mu}D_j^{(i_1)}}+S_{j}-\widetilde{\theta}}\iff e^{\widetilde{\mu}D_j^{(i_1)}}\neq e^{\widetilde{\mu}D_j^{(i_2)}}\iff D_j^{(i_1)} \neq D_j^{(i_2)}.
		\end{aligned}
	\end{equation*}	
	This final condition is equivalent to futures contracts $i_1$ and $i_2$ having different times to maturity, which is assumed here. Finally, direct substitution of $w^*$ into the objective function, $\left(\alpha_0+\alpha_1 w\right)^2+\left(\nu_0+\nu_1 w\right)^2$, yields the optimal objective function value.
\end{proof}

As a corollary to the above proposition we can derive an index value $S_j$ such that there is zero expected squared error.
\begin{corollary}\label{cor:zeroSqrErr}
Under the discrete-time model \eqref{historical_vix_discrete} for the index $S$ with a portfolio of $N=2$ futures contracts on $S$ with maturities $T_{i_1}\neq T_{i_2}$ given by \eqref{discrete_portfolio}, the objective function of Optimization Problem \eqref{optimization_problem} has an optimal value of 0    if and only if 
	\begin{equation}
		S_j=\frac{\beta\widetilde{\mu}\widetilde{\theta}}{\beta\widetilde{\mu}+\bar{r}},\label{cor_sj}
	\end{equation}
	where $\bar{r}:= {(e^{r\Delta t}-1)}/{\Delta t}$.
\end{corollary}
\begin{proof}
	With the notation from Proposition \ref{prop:prop}, the optimal objective function value is 0 for Optimization Problem \eqref{optimization_problem} if and only if $\frac{\alpha_0}{\alpha_1}=\frac{\nu_0}{\nu_1}$. This is equivalent to 
	\begin{equation*}
		\begin{aligned}
			\frac{e^{r\Delta t}-1+\Delta tB_j^{(i_2)}\lambda_j+\beta\mu\Delta t\left(1-\frac{\theta}{S_{j}}\right)}{\Delta t \lambda_j (B_j^{(i_1)}-B_j^{(i_2)})}=\frac{\sqrt{\Delta t}\left(B_j^{(i_2)}-\frac{\beta g(j\Delta t,S_j)}{S_j}\right)}{\sqrt{\Delta t}(B_j^{(i_1)}-B_j^{(i_2)})}.
		\end{aligned}
	\end{equation*}
After rearranging terms and using $\bar{r}= {(e^{r\Delta t}-1)}/{\Delta t}$, we get 
\[\bar{r}S_j+\beta\mu(S_j-\theta)=-\beta g(j\Delta t,S_j)\lambda_j\,.\]
	Now applying equation \eqref{eq:lambda_disc} for the market price of risk $\lambda_j$, the last equation reduces to
	\begin{equation*}
		\begin{aligned}
			\bar{r}S_j+\beta\mu(S_j-\theta)=-\beta \left[{\mu}({\theta}-S_j)-\widetilde{\mu}(\widetilde{\theta}-S_j)\right]\iff \bar{r}S_j=\beta\widetilde{\mu}(\widetilde{\theta}-S_j),
		\end{aligned}
	\end{equation*}which leads us to \eqref{cor_sj} as desired.
\end{proof}

\begin{remark}\label{rmk:zeroSqrErr}
	Interestingly, the critical value is independent of both the trading day $j$ and  maturities of the contracts. \cite{wardIDX} provide the continuous-time analogue of this result with $\bar{r}$ being replaced by the continuously compounded interest rate. This is an intuitive correspondence between the trading frequency and compounding frequency of the interest rate that arises in dynamic   replication. 
\end{remark}

\section{Implementation and Discussion}\label{sec:numerics}
In this section, we implement the optimal strategy described in Section \ref{Optimal_Tracking_Problem}. In particular, we will set the function $g(\cdot,\cdot)$ as $g(t,S_t)=\sigma\sqrt{S_t}$, corresponding to the discrete-time version of the well-known CIR Model (see e.g. \cite{CIR85}). We begin in Section \ref{sec:calibration}, by discussing the calibration of this model under the historical and risk neutral measures. In Section \ref{qual} we simulate the strategy and discuss its properties, comparing and contrasting to VXX.

\subsection{Empirical Estimation}\label{sec:calibration}

The optimal strategy derived in Proposition \ref{prop:prop} requires two sets of parameters: those under the historical measure and those under the risk neutral measure. We begin with the historical measure, where we utilize the Maximum Likelihood Estimation (MLE) framework to derive the parameters $\mu,\theta,$ and $\sigma$. Suppose that the index $S$ follows the CIR model
\begin{align*}
	dS_t = \mu (\theta-S_t) dt + \sigma \sqrt{S_t} dZ_t, 
\end{align*}
where $\mu,\theta,\sigma>0$ and $Z_t$ is a SBM. Further suppose we have observed a discrete sampling of $S_t$ as $\left\{S_{t_1},...,S_{t_n}\right\}$. Then, the conditional probability density of $S_{t_j}$ at time $t_j$ given $S_{t_{j-1}}=s_{j-1}$ with (equally spaced) time increment $\Delta t := t_j-t_{j-1}$ is given by
\begin{align*}
	f^{CIR}&(s_j|s_{j-1}; \theta, \mu, \sigma)\\
	 &= \frac{1}{\tilde{\sigma}^2}\exp\left(-\frac{s_j + s_{j-1}e^{-\mu \Delta t}}{\tilde{\sigma}^2}\right) \left(\frac{s_j}{s_{j-1}e^{-\mu \Delta t}}\right)^\frac{\mathfrak{q}}{2} I_\mathfrak{q} \left(\frac{2}{\tilde{\sigma}^2}\sqrt{s_j s_{j-1} e^{-\mu \Delta t}}\right),
\end{align*}
with the constants
\begin{align*}
	\tilde{\sigma}^2 = \sigma^2 \frac{1-e^{-\mu\Delta t}}{2\mu}, \quad \mathfrak{q} = \frac{2\mu\theta}{\sigma^2}-1,
\end{align*}
and $I_\mathfrak{q}(z)$ is modified Bessel function of the first kind and of order $\mathfrak{q}$.

Using the observed values $\left\{s_j\right\}_{j=1}^n$, the CIR model parameters can be estimated by maximizing the average log-likelihood: 
\begin{align*}
	\ell(\theta,\mu,\sigma | s_1, \dots, s_n) &:= \frac{1}{n-1}\sum_{j=2}^n \ln f^{CIR}(s_j|s_{j-1}; \theta, \mu, \sigma)\\
	&=-2\ln(\tilde{\sigma}) - \frac{1}{(n-1)\tilde{\sigma}^2}\sum_{j=2}^n (s_j + s_{j-1}e^{-\mu \Delta t}) \\
	&\quad + \frac{1}{n-1}\sum_{j=2}^n \left(\frac{\mathfrak{q}}{2} \ln \left(\frac{s_j}{s_{j-1}e^{-\mu \Delta t}}\right)+\ln I_\mathfrak{q} \left(\frac{2}{\tilde{\sigma}^2}\sqrt{s_j s_{j-1} e^{-\mu \Delta t}}\right) \right).
\end{align*}
  After calibrating the model to in-sample data, we obtain   $(\mu,\theta, \sigma)=(10.86, 18.81,  6.37)$.  


Next, we turn to the risk neutral measure. Recalling that the volatility will not change under this measure, we need only calibrate the parameters of the mean reversion: $\widetilde{\mu}$ and $\widetilde{\theta}$. To do so, we employ the Method of Moments (MOM) framework, which is the industry standard. In particular, we will match the conditional expectations $\E^\Q\left[S_{T_i}|S_t\right]$, both across time and across contract maturity. In particular, the futures price on $S$, with maturity $T_i$, as observed at time $t_j$ is given by 
\begin{align*}
	f^{T_i}_t = (S_{t_j}-\widetilde{\theta})e^{-\widetilde{\mu}(T_i-t_j)} +\widetilde{\theta}.
\end{align*}
On any particular trade date, we observe the spot price as well as the entire term structure. Therefore the dataset for calibration of the risk neutral parameters, is the set of triples
\begin{align*}
	\left\{\left(S_{t_j},T_i,f_{t_j}^{T_i}\right):j=1,...,n,i=1,...,N_j\right\},
\end{align*} 
where $N_j$ is the number of futures contracts available for trading on day $j$. We utilize all available price data from all traded VIX futures contracts.  Let us denote $s_j:=S_{t_j}$ and  $f_j^{(i)}:=f_{t_j}^{T_i}$. To implement the MOM for a particular trade date, we consider the loss function
\begin{align*}
	L\left(\widetilde{\mu},\widetilde{\theta}\,\Big|s_j\right):=\frac{1}{2N_j}\sum_{i=1}^{N_j}\left((s_j-\widetilde{\theta})e^{-\widetilde{\mu}T_i} +\widetilde{\theta}-f_j^{(i)}\right)^2.
\end{align*}
To obtain a single set of parameters across all days, we average the loss function over the entire sample and minimize
\begin{align}
	L_{RN}\left(\widetilde{\mu},\widetilde{\theta}\,\Big|s_1,...,s_n\right)&:=\frac{1}{n}\sum_{j=1}^n L\left(\widetilde{\mu},\widetilde{\theta}\,\Big|s_j\right)\notag\\
	&=\frac{1}{n}\sum_{j=1}^n\frac{1}{2N_j}\sum_{i=1}^{N_j}\left((s_j-\widetilde{\theta})e^{-\widetilde{\mu}T_i} +\widetilde{\theta}-f_j^{(i)}\right)^2,\label{eq:riskNeutralOBJ}
\end{align}
by choosing $\widetilde{\mu}$ and $\widetilde{\theta}$. After calibrating to the in-sample data, the estimators for $\widetilde{\mu}$ and $\widetilde{\theta}$ are given as 1.39 and 26.03, respectively. 



\subsection{Performance and Properties of Tracking Strategy}\label{qual}
In this section, we discuss the properties of the optimal dynamic strategy derived in Section \ref{sec:dt_model_vix}. Specifically, we use the parameters computed in Section \ref{sec:calibration}, and simulate the index and its corresponding futures prices via \eqref{historical_vix_discrete} and \eqref{eq:futures_discrete}. In turn, the dynamic portfolio values  and VXX prices are generated by their corresponding strategies. For both, 1-month and 2-month futures are used in the portfolios. Recall that VXX has a linear strategy that starts from 100\% invested in the first-month futures and  linearly decreases the holding in the first-month contract and increases the holding of the second-month contract day by day until the first-month futures expires.

Three different  scenarios are simulated in Figure \ref{fig:vix3SIM}. We consider a 3-contract cycle (63 total days) period and separate the analysis into three cases. In panel (a), $S$ starts at its  historical  long-run level of $\theta$. In panel (b), $S$ starts at a low value of $\theta/3$, and in panel (c)  $S$ starts at a high value of  $3\cdot\theta$. In these scenarios, the index itself exhibits (a) non-directional, (b) upward trending, or (c) downward trending movements. In all cases, we see that the dynamic portfolio is clearly co-moving with the index though the magnitude of   movements differ. The same is true for VXX. Recall that the strategy is designed to track the index returns, rather than the physical price. In other words, the tracking strategy does not account for current or past price  deviations, and only seeks to optimize tracking of the index return for the next period. As a result, the tracking portfolio tends to move in the same direction as the index regardless of their dollar values relative to each other. Over time price deviations  may even accumulate. Nevertheless, as we can see in Figure \ref{fig:vix3SIM}, the dynamic tracking strategy tend to stay closer to the index in all three scenarios.  

To see that the dynamic portfolio tracks the returns of VIX better than VXX, we display a simulated return scatterplot in Figure \ref{fig:VIXscatterSIM} over 6 contract cycles. Panel (a) shows the scatterplot of the returns of the dynamic portfolio vs. the returns of VIX, while panel (b) gives the scatterplot of the returns of VXX vs. the returns of VIX. Both plots have the same $x$-$y$ axis scale for easy comparison. We also superimpose (in black) the line $y=x$. One notices that for the dynamic portfolio, the return pairs are closer to the line $y=x$ indicating that the portfolio returns closely match the returns of VIX. On the other hand, for VXX the returns are further from this line and the slope appears to be less than 1 in value. In fact, both statements are statistically significant. 

Our linear regressions to the two scatterplots yield the following results: In  (a), the slope is 1.0030 with an standard error of $1.6408\cdot10^{-3}$,   the intercept is $-8.5514\cdot10^{-4}$ with a standard error of  $1.2979\cdot10^{-4}$, and   $R^2= 99.97\%$. In (b), the slope is 0.8734 with an standard error of $2.9477\cdot10^{-3}$,   the intercept is $-10.1524\cdot10^{-4}$ with a standard error of  $2.3318\cdot10^{-4}$, and   $R^2=  99.86\%$. We obtain a  p-value of $0.078$ for testing the null hypothesis $H_0:\left\{\text{slope}=1\right\} $ for the regression of the dynamic portfolio returns vs. VIX returns in panel (a), which is greater than many standard significance levels. The corresponding test for VXX has a p-value of $6.521\cdot10^{-76}$.  Finally, we caution that the high $R^2$ value (over 99\%) implies that the dynamic portfolio can generate very similar returns to the VIX on a daily basis, but does not mean that it can perfectly replicate the physical dollar value of VIX over a long period of time.  


As for the strategy, we give an example in Figure \ref{fig:vixSTRATsim} over two contract cycles. Interestingly, the dynamic strategy is quite linear falling from a little over 2 in value to 1 by maturity. This indicates that the strategy is also short \%100 of its value in 2-month futures rising to 0\% by maturity of the 1-month futures. For comparison, we also plot the strategy of VXX, which is linear from 1 to 0 over each contract cycle.

\begin{figure}[h]
	\begin{centering}
		\subfigure[$S_0=\theta$]{\includegraphics[trim={0.8cm 1.6cm 0.8cm 1.6cm},clip,width=4.5in]{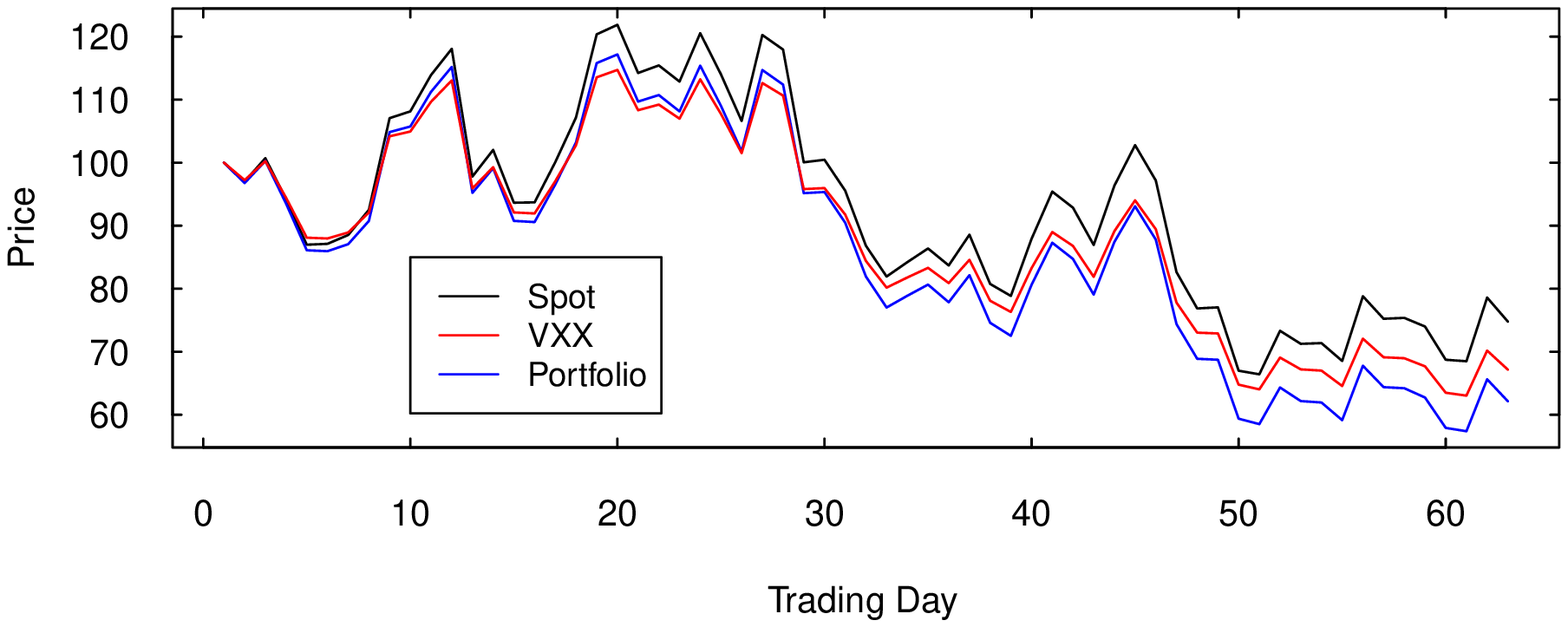}}
		\subfigure[$S_0=\frac{1}{3}\cdot\theta$]{\includegraphics[trim={0.8cm 1.6cm 0.9cm 1.6cm},clip,width=4.5in]{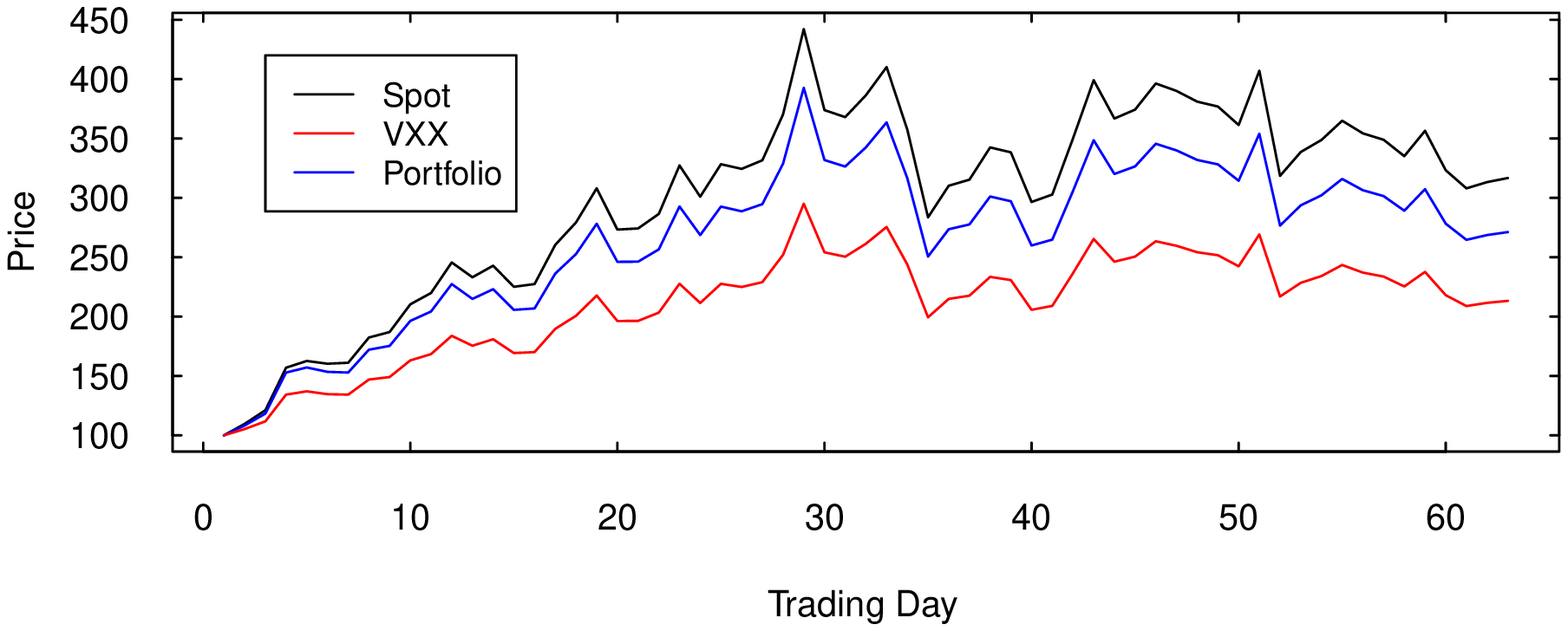}}
		\subfigure[$S_0={3}\cdot\theta$]{\includegraphics[trim={0.8cm 1.6cm 0.9cm 1.6cm},clip,width=4.5in]{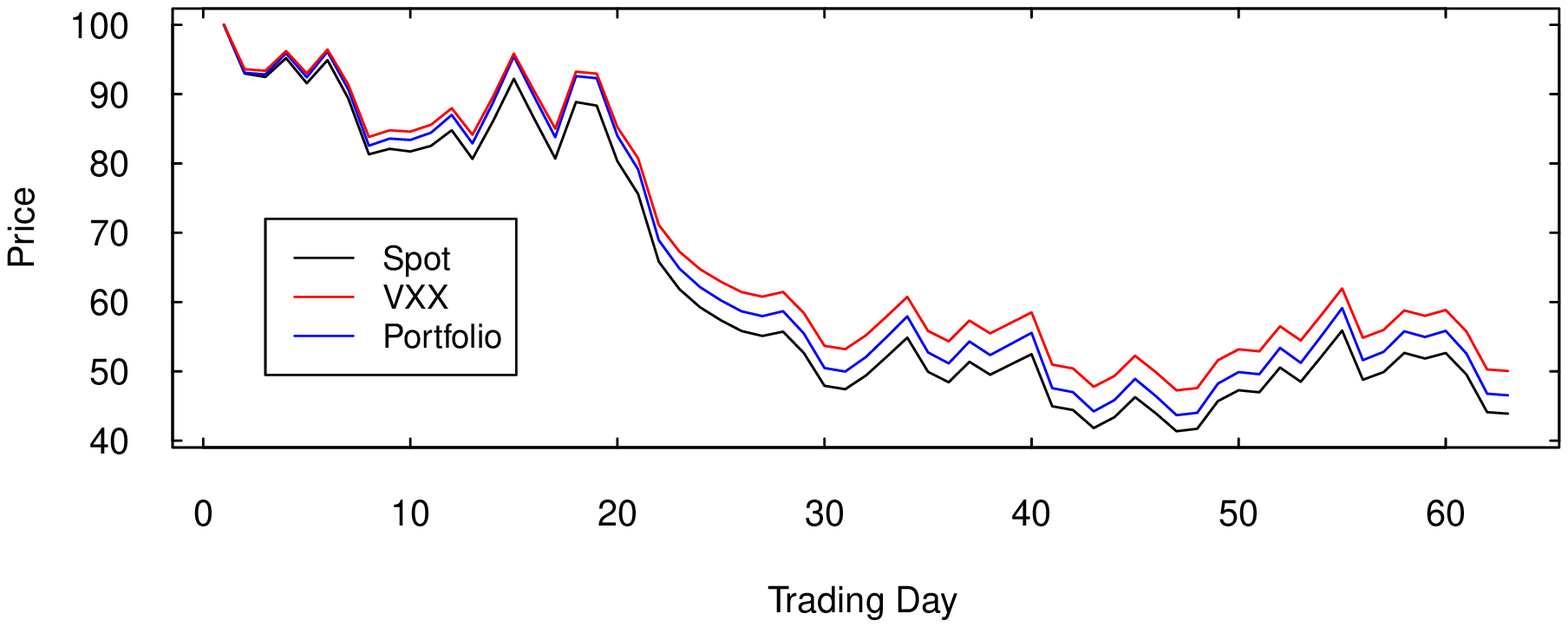}}
		\caption{\small{Sample paths of the index, VXX, and  dynamic portfolio under 3 different scenarios: (a) when VIX starts at its historical long-run mean $\theta$ (b) when VIX stats at a lower value $\theta/3$, and (c) when VIX stats at a higher value $3\theta$. All portfolios are normalized to start at \$$100$ and  their time series are plotted over 3 contract cycles. The parameters are $\mu=10.86$, $\theta=18.81$, $\sigma=6.38$, $\widetilde{\mu}=1.39$, and $\widetilde{\theta}=26.03$.  The x-axis marks the trading day, while the y-axis marks the price.  }}
		\label{fig:vix3SIM}
	\end{centering}
\end{figure}

\clearpage
\begin{figure}[t]
	\begin{centering}
		\includegraphics[trim={0.9cm 1.6cm 0.9cm 1.6cm},clip,width=3.2in]{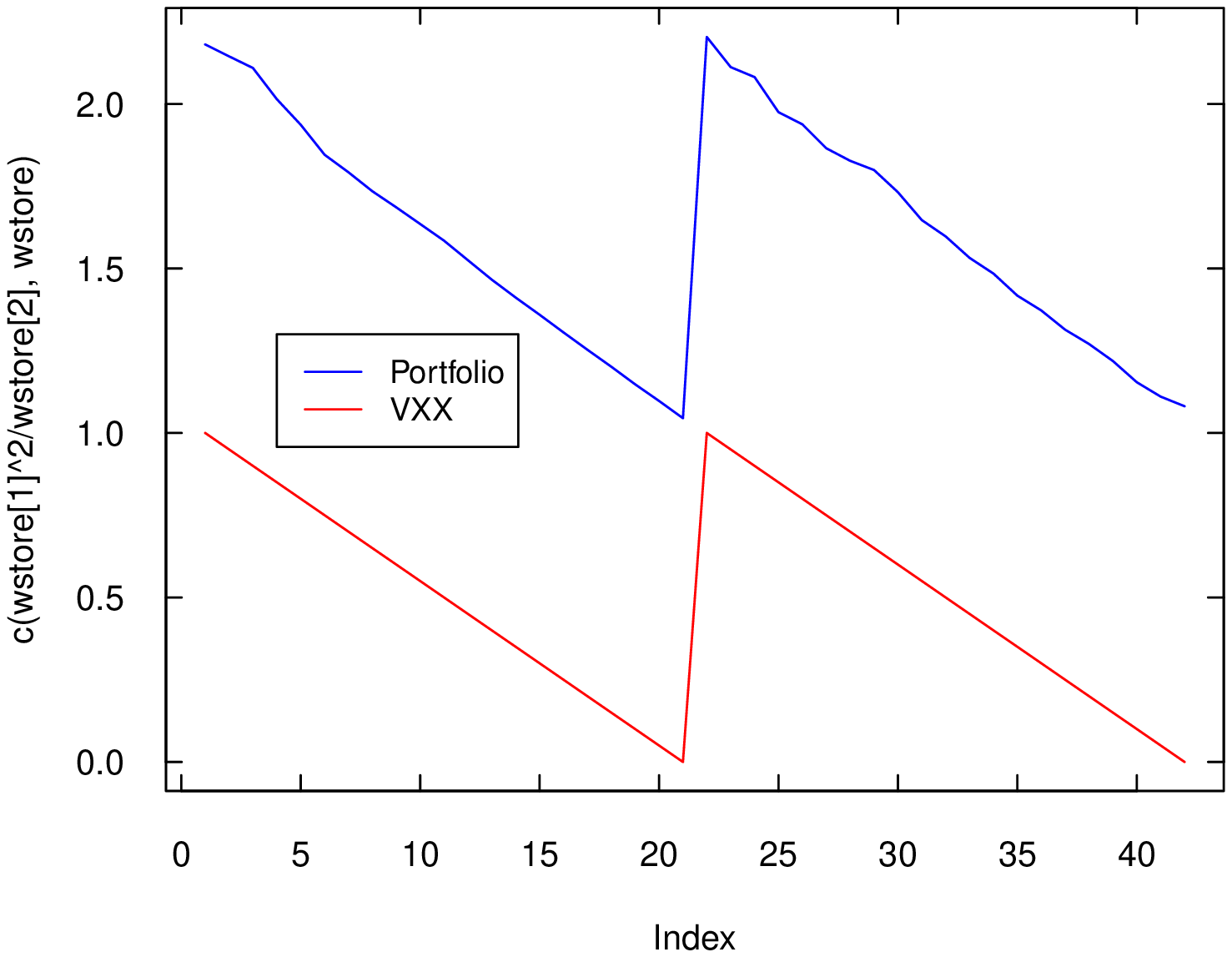}
		\caption[Dynamic Strategy for Tracking VIX]{\small{Sample paths of the portfolio weight, expressed in terms of the fraction of wealth invested in  the 1-month futures, associated with the dynamic portfolio and VXX over  42 trading days (i.e. 2 contract cycles). The rest of the portfolio is invested in the 2-month futures. The parameters are $\mu=10.86$, $\theta=18.81$, $\sigma=6.38$, $\widetilde{\mu}=1.39$, $\widetilde{\theta}=26.03$ and $S_0=\theta$. The x-axis marks the trading day, while the y-axis marks the weight on 1-month futures. }}
		\label{fig:vixSTRATsim}
	\end{centering}
\end{figure}
\begin{figure}[b]
	\centering
	\subfigure[Dynamic Portfolio]{\includegraphics[trim={0.5cm 1.4cm 0.5cm 1.4cm},clip,width=2.8in]{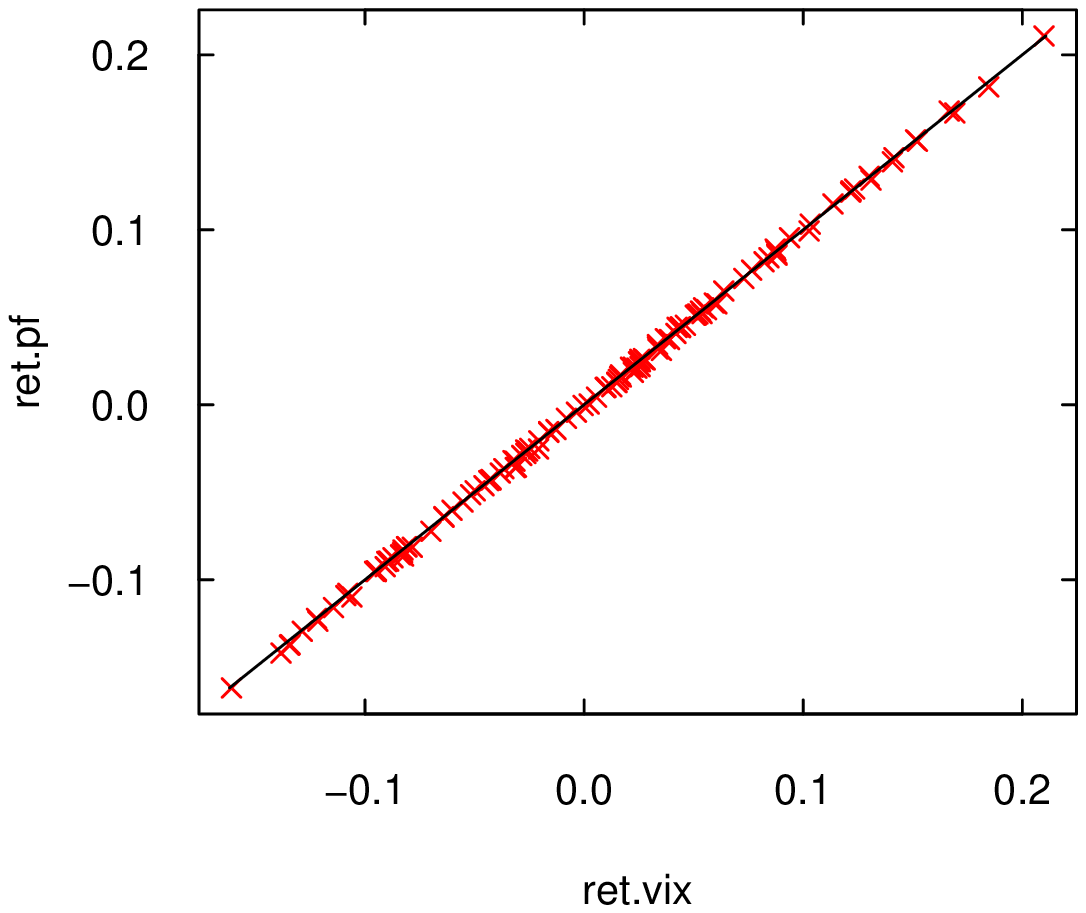}}
	\subfigure[VXX]{\includegraphics[trim={0.5cm 1.4cm 0.5cm 1.4cm},clip,width=2.8in]{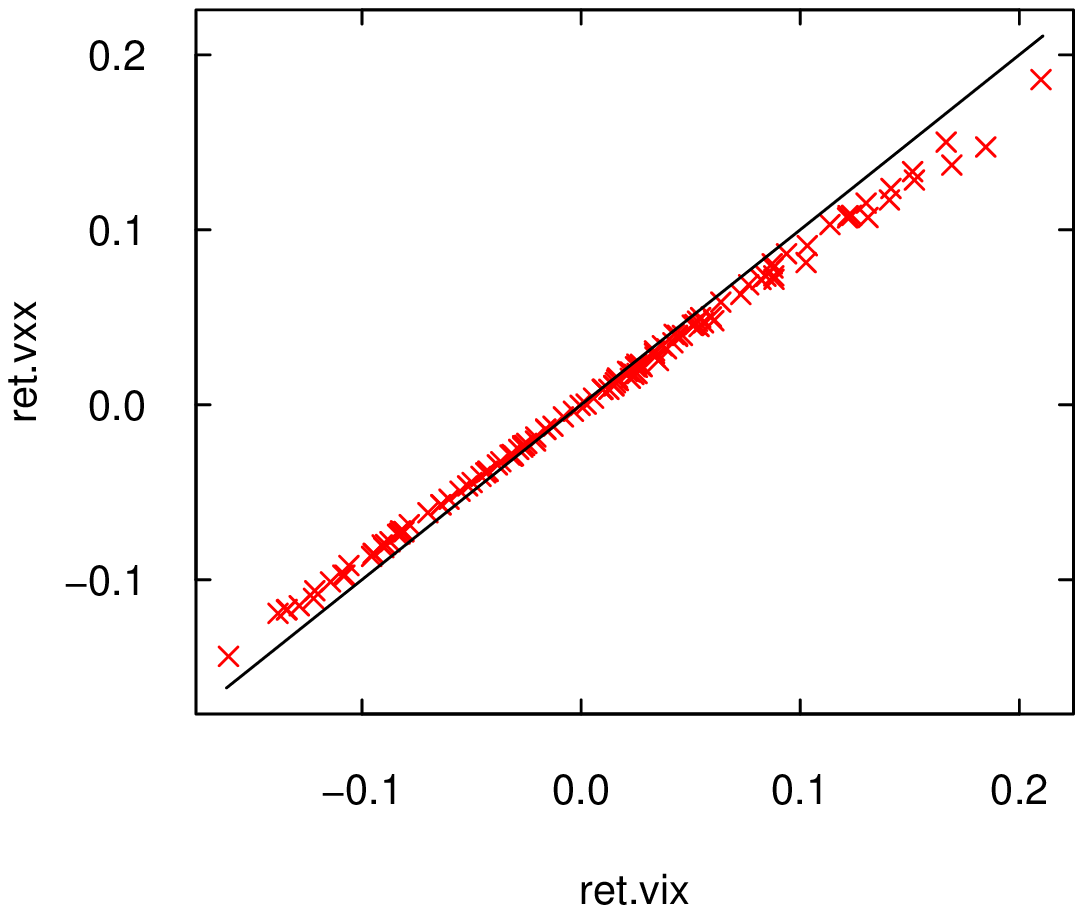}}
	\caption{\small{Scatterplots of returns for (a) dynamic portfolio vs.  VIX, and (b) VXX vs. VIX. We also plot a solid straight line with zero intercept to illustrate the   deviations of the dynamic portfolio or VXX returns from the VIX returns. The parameters are $\mu=10.86$, $\theta=18.81$, $\sigma=6.38$, $\widetilde{\mu}=1.39$, $\widetilde{\theta}=26.03$, and $S_0=\theta$. Trading is over 126 trading days (6 contract cycles). The x-axis marks the index returns (in decimals), while the y-axis marks the dynamic portfolio/VXX returns (in decimals).}}\label{fig:VIXscatterSIM}
\end{figure}

\clearpage

\section{Concluding Remarks}\label{sect-conclude}
In this paper we have presented the construction of static and dynamic portfolios of VIX futures that are optimized for VIX tracking. As discussed from both empirical and theoretical perspectives, futures prices tend not to react quickly enough to movements in the spot VIX. They also bring persistent negative returns and thus makes tracking VIX difficult.  In view of the pitfalls of static portfolios, we design and implement a dynamic trading strategy that adjusts futures positions daily to optimally replicate VIX's returns. The model is calibrated to historical data and a simulation study is performed to understand the properties exhibited by the strategy. In addition, comparing to the volatility ETN, VXX, we find that our dynamic strategy has a superior tracking performance.

There are a number of directions for future research. The question of how well ETFs or ETNs track their reference index is not limited to VIX-based products. It is relevant and important for all other asset classes, ranging from equities to commodities.\footnote{See \cite{guoleung} for a study on the tracking performance of commodity ETFs.} In our static replicating portfolios, we have chosen to include futures only for comparison to VXX.  Alternatively, one can also construct a static portfolio of options and optimize to minimize tracking errors with respect to an index or its leveraged exposure (see \cite{LeungLorig2015}). It would be interesting to compare the tracking performance of portfolios with different derivatives. Strategies and instruments aside, a major challenge in all tracking problems is to properly model the evolution of the index. Hence, accurate predictions of the index movements can potentially enhance   tracking performance. To that end, one future research direction is to apply machine learning techniques to incorporate predictive analytics into the tracking problem.

\begin{small} 
	\bibliographystyle{apa}
	\bibliography{mybib}
\end{small}

\end{document}